\documentclass[envcountsect,envcountsame,runningheads]{llncs}
\usepackage[T1]{fontenc}
\usepackage[english]{babel}
\usepackage{graphicx}
\usepackage{color}
\usepackage{hyperref}

\spnewtheorem{ldefin}[theorem]{Definition}{\bfseries}{\rmfamily}
\spnewtheorem{lconstr}[theorem]{Construction}{\bfseries}{\rmfamily}
\spnewtheorem{lnotation}[theorem]{Notation}{\bfseries}{\rmfamily}
\spnewtheorem{fact}[theorem]{Fact}{\bfseries}{\rmfamily}
\usepackage{multirow}%
\usepackage{amsmath,amssymb,amsfonts}%
\usepackage{mathrsfs}%
\usepackage[title]{appendix}%
\usepackage{xcolor}%
\usepackage{textcomp}%
\usepackage{manyfoot}%
\usepackage{booktabs}%
\usepackage{algorithm}%
\usepackage{algorithmicx}%
\usepackage{algpseudocode}%
\usepackage{listings}
\usepackage{tikz}
\usepackage{tikz-qtree}
\usetikzlibrary{decorations.text,calc,arrows.meta}
\usetikzlibrary{shapes.geometric, arrows}
\usepackage{mwe}
\usepackage{marginnote}
\usepackage[capitalise]{cleveref}
\usepackage{enumitem}
\usepackage{bm}

\usepackage{mathtools}
\usepackage{xfrac}
\usepackage{stmaryrd}
\crefname{subsection}{subsection}{subsections}
\newcommand{\cond}{>}

\newcommand{\less}{\preccurlyeq}
\newcommand{\discond}[1]{\,{\cond}^{#1} }
\newcommand{\disless}[1]{\,{\less}^{#1} }

\newcommand{\pr}[1]{\bm{\mathsf{Pr}}(#1)}

\newcommand{\fe}{\models^\exists}
\newcommand{\fum}[1]{\models^\forall_{#1}}
\newcommand{\fem}[1]{\models^\exists_{#1}}

\newcommand{\atm}{\mathsf{Atm}}

\newcommand{\sph}[1]{S(#1)}
\newcommand{\model}{\mathcal{M}}
\newcommand{\val}[1]{\llbracket#1 \rrbracket}

\newcommand{\cpcount}{\emph{cp}-counterfactual\xspace}
\newcommand{\cpcounts}{\emph{cp}-counterfactuals\xspace}

\newcommand{\card}[1]{|#1|}

\newcommand{\weight}[2]{\mathbf{w}_{#1}(#2)}

\newcommand{\preorder}[1]{\preceq_{#1}}

\newcommand{\bleft}{\mathrel{\mathpalette\bleftinn\relax}}
\newcommand{\bleftinn}[2]{%
  \ooalign{%
    \raisebox{.2ex}{$#1\blacktriangleleft$}\cr
    $#1\leq$\cr
  }%
}

\newcommand{\gtruthset}[2]{\llbracket#1 \rrbracket_{#2}}

\newcommand{\disset}[3]{\mathsf{D}(#1,#2)^{#3}}
\newcommand{\agset}[3]{\mathsf{A}(#1,#2)^{#3}}

\newcommand{\upmod}[3]{\mathcal{M}^{#1/#2}_{#3}}

\newcommand{\upw}[3]{W^{#1/#2}_{#3}}
\newcommand{\ups}[3]{S^{#1/#2}_{#3}}

\newcommand{\upv}[3]{V^{#1/#2}_{#3}}

\newcommand{\upsph}[2]{S^{\Gamma/x}_{#1}(#2)}
\newcommand{\upds}[3]{S^{{#1}/{#2}}_{\mathsf{#3}}(#2)}

\newcommand{\modelsm}[1]{\models_{#1}}

\newcommand{\prel}{\preorder{lex}}
\newcommand{\preil}{\preorder{xel}}
\newcommand{\presl}{\prec_{lex}}
\newcommand{\presil}{\prec_{xel}}

\newcommand{\Lcp}{\mathcal{L}^{\discond{}}_{cp}}
\newcommand{\Ll}{\mathcal{L}^{\discond{}}}
\newcommand{\Lcpp}{\mathcal{L}^{\disless{}}_{cp}}
\newcommand{\Llp}{\mathcal{L}^{\disless{}}}

\newcommand{\xelleft}[1]{\bleft_{#1}}

\newcommand{\xellefts}[1]{\blacktriangleleft_{#1}}

\newcommand{\lexleft}[1]{\trianglelefteq_{#1}}

\newcommand{\lexlefts}[1]{\vartriangleleft_{#1}}

\newcommand{\eq}[1]{=_{#1}}

\newcommand{\cplless}[1]{cp^{\less}(#1)}

\setbox1=\hbox{$\bullet$}\setbox2=\hbox{\tiny$\bullet$}

\newcommand{\VWu}[1]{\mathsf{VW}^\mathsf{#1}}
\newcommand{\VCu}[1]{\mathsf{VC}^\mathsf{#1}}

\newcommand{\VW}{\mathsf{VW}}
\newcommand{\VC}{\mathsf{VC}}

\newcommand{\setdiff}{\!\setminus\!}

\newcommand{\numdate}{\the\day/\the\month/\the\year}

\newcommand{\sep}{,~\!\!}
\renewcommand{\upmod}[3]{\mathcal{M}^{#1\sep#2}}
\renewcommand{\upw}[3]{W^{#1\sep#2}}
\renewcommand{\ups}[3]{S^{#1\sep#2}}
\renewcommand{\upv}[3]{v^{#1\sep#2}}
\renewcommand{\upsph}[2]{S^{\Gamma\sep x}(#2)}

\newcommand{\modelsworld}{\Vdash}
\renewcommand{\fum}[1]{\modelsworld^\forall}
\renewcommand{\fem}[1]{\modelsworld^\exists}
\renewcommand{\modelsm}[1]{\modelsworld}
\newcommand{\validm}[1]{\models_{#1}}

\begin{document}
\title{A significance-based account of \emph{ceteris paribus} counterfactuals
}
%
%\titlerunning{Abbreviated paper title}
% If the paper title is too long for the running head, you can set
% an abbreviated paper title here
%
\author{Avgerinos Delkos\inst{1}
%\orcidID{0000-1111-2222-3333} 
\and
Marianna Girlando\inst{2}
%\orcidID{1111-2222-3333-4444}
}
\authorrunning{Delkos and Girlando}
% First names are abbreviated in the running head.
% If there are more than two authors, 'et al.' is used.
%
\institute{
University of Birmingham, UK
\and
University of Amsterdam, Netherlands
%\email{lncs@springer.com}\\
%\url{http://www.springer.com/gp/computer-science/lncs} 
%\email{\{abc,lncs\}@uni-heidelberg.de}
}
\maketitle              % typeset the header of the contribution

\begin{abstract}
When evaluating a counterfactual statement, it is often convenient to specify conditions that ought to be kept unchanged. 
Formally, this can be done by associating to each counterfactual a  \emph{ceteris paribus} set of formulas, specifying the facts that ``ought to be kept unchanged''. 
%associating to a counterfactual $A \cond B$ a set of formulas $\Gamma$, which 
%Counterfactuals can then be assessed at hypothetical states where ``ought to be kept unchanged'' (\emph{ceteris paribus}). 
\emph{Ceteris paribus} counterfactuals originate in the debate between D. Lewis and Fine in the 1970s, and have been captured in formal accounts. However, these accounts are merely based on `counting' formulas, and can yield counterintuitive results. 
In this paper, we develop a novel approach to evaluate \emph{ceteris paribus} counterfactuals at (weakly) centered sphere models, by taking into account the `significance' of formulas that ought to be kept unchanged. Hypothetical states that keep the most significant formulas unchanged will be prioritized in the evaluation of a counterfactual. 
We show that the resulting notion of validity  coincides with theoremhood in Lewis' conditional logics $\VC $ or $\VW$.

%\keywords{First keyword  \and Second keyword \and Another keyword.}
\end{abstract}

%%%%%%%%%%%%%%%%%%%%%%%%%%%%%%%%%%%%%%%%%%%%
%%%%%%%%%%%%%%%%%%%%%%%%%%%%%%%%%%%%%%%%%%%%
\section{Introduction}
\label{sec:introduction}
%%%%%%%%%%%%%%%%%%%%%%%%%%%%%%%%%%%%%%%%%%%%
%%%%%%%%%%%%%%%%%%%%%%%%%%%%%%%%%%%%%%%%%%%%

\newcommand{\nc}{\emph{nc}\xspace}

Counterfactuals are sentences allowing to 
reason about hypothetical states of affairs. 
For instance: ``{If Nixon had pressed the button, there would have been a nuclear holocaust}''. We will use this counterfactual, due to Fine~\cite{this-is-fine} and known as \emph{Nixon counterexample} (\nc), as our running example.  
Several classes of possible-world models have been introduced to interpret  counterfactuals, among which \emph{sphere models} 
%\cite{lewis1973}
and \emph{preferential models} \cite{lewis1973,burgess1981quick}, both encoding notions of similarity between states. To evaluate a counterfactual $A \discond{} B$ in a world $x$, one checks whether $B$ holds in states that satisfy $A$ and which are most similar to $x$.  
%\todoavg{24/02 merged ref above}

\emph{Ceteris paribus} counterfactuals, with \emph{ceteris paribus} (\emph{cp} for short) meaning ``all other things being equal'', specify additional conditions to be taken into account when evaluating a counterfactual. 
\emph{Cp}-counterfactuals were 
%informally 
introduced 
by D. Lewis in~\cite{lewis1979counterfactual} as an answer to Fine's critique of sphere models in~\cite{this-is-fine}. 
Fine observed that, in relation to (\nc) above, 
a world in which a small miracle prevents the launch of the missile is more similar to the actual world than a world where a nuclear holocaust occurs. Hence, the counterfactual (\nc) is counterintuitively evaluated as 
false, because at the state closest to the actual world and at which Nixon pushes the button there is no nuclear holocaust (but a small miracle occurs).   
Lewis replied that, to get the intuitively correct evaluation, (\nc) needs to be changed into a \cpcount:  ``If Nixon had pressed the button, \emph{and no miracle occurred}, there would have been a nuclear holocaust''. 
While the  similarity ordering among possible states is left intentionally vague in Lewis' account, the idea that specifying an appropriate context is crucial in the evaluation of counterfactuals has received widespread attention in the literature. 
%Thus, the similarity ordering among worlds is something that can be updated depending on the context, as observed in ~\cite[p. 465]{lewis1979counterfactual}: 
%\begin{quote}
%Overall similarity among worlds is some sort of resultant of similarities and differences of many different kinds, and I have not said what system of
% weights or priorities should be used to squeeze these down into a single relation of overall similarity. I count that a virtue.
% Counterfactuals are both vague and various. Different resolutions of the vagueness of overall similarity are appropriate
% in different contexts.
%\end{quote}
%Trying to match Lewis' intuition about context-dependence of the similarity ordering, several formal accounts of \emph{cp}-counterfactuals were developed in the literature. 
Most notably, theories of \cpcounts based on selection-function semantics were developed in  ~\cite{priest2008introduction,weiss2017semantics}, while in~\cite{von1963logic,von1972logic} a \emph{cp}-preferential operator is introduced, an approach later formalised in~\cite{van2009everything}\footnote{ 
%Possible-world semantics is not the only way to formalise \emph{cp}-sentences. 
To mention another approach, \emph{premise semantics} also allows to formalise \emph{cp}-reasoning, by establishing causal links between sentences.
%at play in the evaluation of a counterfactual. 
Refer, e.g., to 
\cite{goodman1947problem,kratzer1981partition,kaufmann2013causal}. 
}.

Our contribution 
stems from 
the works of P. Girard and Triplett, who defined in \cite{DBLP:journals/corr/GirardT16,girard2018prioritised} an analysis of \emph{cp}-counterfactuals within  centered preferential models. 
They introduce a \emph{cp}-counterfactual $[A,\Gamma]B$,  
to be understood as Lewis' counterfactual operator $A \discond{} B$ where the truth values of formulas in $\Gamma$  ought to be kept unchanged when evaluating the counterfactual. 
%Intuitively, referring to \eqref{eq:nixon}, the sentence ``no miracle occurs'' would be included in $\Gamma$. 
%Then,  in \cite{DBLP:journals/corr/GirardT16,girard2018prioritised}  
The formula $[A,\Gamma]B$ can be evaluated at a world $x$ by `updating' the preorder relation of preferential models in three different ways. 
In their first proposal, which we call \emph{strict evaluation}, only worlds which agree with $x$ on all formulas in $\Gamma$ are considered for the evaluation. This operation  amounts to check satisfiability of $(A \wedge \bigwedge \Gamma)  \discond{} B$ at $x$. 
However, several worlds might fail to satisfy $A \wedge \bigwedge \Gamma$, thus trivializing the result.  
%\todoavg{24/02 not always all, if the real world satisfies $A$. Pedantic but possible}
This can be remedied by relaxing the strict evaluation 
into 
two `prioritisations'~\cite{DBLP:journals/corr/GirardT16,girard2018prioritised}: \emph{na\"ive counting}
%in which worlds are ranked according to the cardinality of the `agreement set', that is, the subset of formulas in $\Gamma$ over which a world agrees with the actual world $x$, 
and \emph{maximal supersets},
%, in which the subset relation over agreement sets is taken into account. 
which ultimately rely on `counting' the number of formulas in $\Gamma$  satisfied by the worlds. However, it is easy to find examples of \cpcounts whose evaluations under the aforementioned relaxed prioritisations yield counterintuitive results (\cref{remark:critique_to_girard}). 

%\todoavg{24 slightly rewrote above paragraph}

In this paper, we introduce the  \emph{disagreement update}, which takes into account a notion of  `significance'
of formulas in $\Gamma$ to re-arrange worlds in a system of spheres. Intuitively, worlds which `agree the most' with the actual world over the most significant formulas in $\Gamma$ will be moved `closest' to the actual world, and will thus make a difference for the evaluation. 
Formally, we shall first associate to each formula 
a \emph{weight}, 
measuring its implausiblity w.r.t. the actual state. This notion, generalised to sets of formulas, allows us to identify the significant formulas for the evaluation of a \cpcount.  
%Significance of formulas is then defined by means of weights. 
%By measuring the weight of specific sets of formulas associated to each world, we select the worlds which should be relevant for the evaluation, which will be the ones satisfying the largest number of most significant formulas. 
The system of spheres is updated by taking the weight  of specific sets of formulas into account. The counterfactual is then evaluated in the updated model. This fine-grained update, which does not solely rely on counting formulas,  allows us to overcome the problems encountered when evaluating a \emph{cp}-counterfactual with the prioritarisations from~\cite{DBLP:journals/corr/GirardT16,girard2018prioritised}.

Our update gives rise to two logics, $\VWu{d}$ and $\VCu{d}$, which are \emph{ceteris paribus} versions of Lewis' logics $\VW$ and $\VC$, respectively characterised by weakly centered and centered sphere models~\cite{lewis1973}\footnote{Our account is based on sphere models, as these allow for an intuitive definition of weights. However, our definitions can easily be adapted to preferential structures.}. We  %compare our approach with the proposal from \cite{DBLP:journals/corr/GirardT16,girard2018prioritised}, and we conclude by proving completeness of (most of) our 
show that validity in $\VWu{d}$ and $\VCu{d}$ coincides with validity in Lewis' logics $\VW $ and $\VC$.  

The paper is structured as follows. In \Cref{Preliminaries} we provide the basic notions, and in  
\cref{sec:weight} we introduce  weight and significance of literals. \Cref{sec:update} presents the disagreement update and compares it with the prioritisations from \cite{DBLP:journals/corr/GirardT16,girard2018prioritised}. 
%and introduces the corresponding notion of satisfiability of formulas.
% Then, \Cref{sec:notions of ceteris paribus} discusses the relations between the different updates, and
Then, \Cref{sec:completeness} proves soundness and completeness of our logics, while
\Cref{sec:conclusions} concludes with some directions for future work.

%%%%%%%%%%%%%%%%%%%%%%%%%%%%%%%%%%%%%%%%%%%%
%%%%%%%%%%%%%%%%%%%%%%%%%%%%%%%%%%%%%%%%%%%%
\section{Preliminaries}
\label{Preliminaries}
%%%%%%%%%%%%%%%%%%%%%%%%%%%%%%%%%%%%%%%%%%%%
%%%%%%%%%%%%%%%%%%%%%%%%%%%%%%%%%%%%%%%%%%%%
Given a countable list of propositional atoms $\atm = \{p_0, p_1, p_2, \dots\}$, the formulas of the language $\Lcp$ are generated
as follows:
%\footnote{
%In this na\"ive (but convenient!) definition,  $\Gamma$ might be non-wellfounded. The full definition, avoiding this issue, is reported in the Appendix (Definition~\ref{def:full:language}). 
%Following~\cite{girard2018prioritised}, 
%}: 
$
 A ::= p \mid 
 \bot 
 \mid A \rightarrow A
 \mid A\discond{\Gamma}A
 $. 
We set $\neg A := A \rightarrow \bot$, and we can standardly define the other propositional connectives. 
For simplicity, we often write $\bar p$ instead of $\neg p$. We call \emph{literals} the atomic formulas and their negations. We set $\bar {\bar p} = p$. 
The operator $\discond{\Gamma} $ is the \emph{cp}-counterfactual conditional, with $\Gamma$, the \emph{cp}-set, being a finite set of literals of $\Lcp$
\footnote{We could have easily allowed $\bot$, $\rightarrow$-formulas and $\discond{\Delta}$-formulas with $\Delta = \{\varnothing\}$ in \emph{cp}-sets. We chose not to do it, to simplify the definitions in \cref{sec:weight,sec:update}.}. We take $\Gamma$ to be \emph{paired}, that is:  for any atom $p$, it holds that $p \in \Gamma$ iff $\bar p \in \Gamma$. 
A formula  $A\discond{\Gamma} B $ reads ``If $A$ were the case, and all things in $\Gamma$ were unchanged, then $B$ would have been the case''.
%Intuitively, $\Gamma$ dictates what ought to be kept unchanged when reasoning about different worlds. 
When $\Gamma$ is empty, 
$A\discond{\Gamma} B$ corresponds to Lewis' 
counterfactual conditional, and we write it as $A\discond{} B$. 
%omit specifying the superscript. 
%Thus, $A\discond{} B $ reads ``If $A$ were the case, then $B$ would have been the case''. 
We denote by $\Ll$ the language of Lewis' logics, obtained by setting   $\Gamma = \{\varnothing\}$ in $\Lcp$. 

%\todomar{check if we need the comparative plausibility here; it can go in the appendix}
%Lewis introduced a second modal operator, the \emph{comparative plausibility}, denoted by $A \disless{} B$ and interpreted as ``$A$ is at least as plausible as $B$''
%\footnote{Equivalently, the operator $A \disless{} B$ can be read as ``$A$ is no more implausible than $B$''. }.
%The counterfactual and the comparative plausibility operators are interdefinable. In \Cref{sec:completeness} we shall introduce a \emph{ceteris paribus} version of the comparative plausibility operator: $A \disless{\Gamma} B$, meaning ``All things in $\Gamma$ being equal, $A$ is at least as plausible as $B$''. We shall prove that, similarly as in Lewis' setting, $A \discond{\Gamma} B$ is definable in terms of $A \disless{\Gamma} B$. The completeness results in \cref{sec:completeness} rely on the \emph{ceteris paribus} comparative plausibility operator. 

%\begin{remark} 
By requiring $\emph{cp}$-sets to be paired, we guarantee that both a literal and its negation are taken into account for the evaluation, following the intuition that, if $p$ should be `kept unchanged', the same should hold for $\bar p$. 
%\todomar{commented sentence about remark}
%In \cref{remark:paired} we provide an example where dropping this requirement provides counterintuitive evaluations. 
Moreover, since the worlds are classical, we have that a world always satisfies some formula in a given \emph{cp}-set.  
We now introduce sphere models, from~\cite{lewis1973}.

\begin{definition}
\label{def:sphere model}
A \emph{sphere model} $\model = \langle W, S, v\rangle$ is composed of a non-empty set of worlds $W$, a valuation function $v: \atm \longrightarrow \mathcal{P}(W)$, and a sphere function $S: W\longrightarrow \mathcal{P}(\mathcal{P}(W))$ which associates to each world $x \in W$ a \emph{system of spheres} $\sph{x}$. We denote by $\alpha, \beta, \dots$ the elements of $\sph{x}$, called \emph{spheres}.  $S$ satisfies the properties of \emph{non-emptiness}: For any $x \in W$, $\{\varnothing\} \notin \sph{x}$; %\marianna{changed non-emptiness, before it was $\varnothing \notin \sph x$}
and \emph{nesting}: For any $x \in W$, $\alpha, \beta \in \sph{x}$, either $\alpha \subseteq \beta$ or $\beta \subseteq \alpha$. 
Then, a \emph{weakly centered sphere model} is a sphere model satisfying \emph{weak centering}: for any $x \in W$, for any $\alpha \in \sph{x}$, $x \in \alpha$. 
A \emph{centered sphere model} is a sphere model satisfying \emph{centering}, that is, for any $x \in W$, $\{x\} \in \sph{x}$\footnote{Due to nesting,  at any centered sphere model it holds that for any $\alpha \in \sph{x}$, $\{x\} \subseteq \alpha$. 
%Observe that any centered model is also weakly centered.
}. 
\end{definition}

\begin{definition}
\label{def:sat_atomic}
    The satisfaction relation of literals at a world $x$ of $\model$ is defined as 
     $\model, x \modelsm{\mathsf{d}} p $ \emph{iff} $x\in v(p)$ and 
     $\model, x \modelsm{\mathsf{d}} \bar p $ \emph{iff} $x\notin v(p)$.
\end{definition}

%\footnote{The conditional logics introduced in \cite{lewis1973} all have the finite model property. We conjecture that the same holds for the logics we shall introduce, but for the sake of simplicity we only treat the finite case in this paper.}. 
%Without loss of generality, we assume our models to be finite \cite{lewis1973}. 
%Any centered 
%model is also weakly centered. 
In this paper, we restrict our analysis to finite sphere models. 
We shall define the satisfaction of full formulas in (weakly) centered sphere models in \cref{sec:update}, after introducing the elements needed to evaluate \emph{cp}-counterfactuals.

\section{
%Measuring significance
From weight to significance
}
\label{sec:weight}
%\todoavg{23/02 since this has been compacted, i propose new title: From weight to significance}
%\marianna{23 good point, but I'd just go for Weights and Significance}
%\marianna{this needs to be adjusted}
To evaluate $cp$-counterfactuals, we shall refine the 
prioritisations introduced in \cite{DBLP:journals/corr/GirardT16,girard2018prioritised} by differentiating literals w.r.t. 
their significance. 
To illustrate what we mean by `significance', let $h$ be a formula stating that a tossed coin lands on the heads side, and $f$ be a formula expressing some law of physics. While satisfying or failing to satisfy $h$ represents a small change in the state of the world, a world where $c$ does not hold would be very different from a world where $c$ is true.  Consequently, when considering things that `ought to remain unchanged' $c$ is a much more 
relevant, or significant, formula than $h$.
%
%two formulas $h,c$ meaning ``if you toss a coin it will land heads'' and
%``watermelons are berries''
%``cold air is heavier than hot air ''
%respectively. 
%Satisfying $h$ versus $\bar h$ would represent a small change in the state of the world (it is easy to imagine a world with either heads or tails as the outcome of the toss), while a world where $\bar c$ holds would be very different from a world where $c$ is the case. In a world satisfying  $\bar c$, the laws of physics would be quite different from those in our actual world and hence, such a world  would be very different from ours. 
%Consequently, when considering things that `ought to remain unchanged' $c$ is a much more 
%relevant 
%candidate than $h$, i.e.\ $c$ will be is more significant than $h$. 
%
%\todoavg{Changed "the sky is blue" due to reviewer's comment.}
%
We shall formalize this intuition by first associating to each formula a  \emph{weight}, measuring how (im)plausible a formula is in a system of spheres. 
%Then, we will show how the comparison of weights naturally induces a notion of significance of formulas (\Cref{subsec:weights:sign}). 
%The (simpler) notion of weight will be the basis of most of the definitions employed in this work. Most notably, we will define our updates by looking at the weights of specific sets of formulas. 
%Then, in \Cref{subsect:weight:sets}, we generalize this notion to sets of formulas. 
We will then generalise this notion to sets of formulas, and illustrate the role of significant formulas on evaluating \cpcounts. 

%\todoavg{22/02 added ``on evaluating \cpcounts'' above}

For this section, let us fix a  (weakly) centered model $\model = \langle W, S, v\rangle$, some $x \in W$ and a system of spheres $\sph{x}$.  
Thanks to nesting and (weak) centering, $\sph{x}$ can be ordered w.r.t set inclusion:  $\sph{x} = \{\alpha_0, \dots, \alpha_n \}$ with $\alpha_i \subseteq \alpha_{i+1}$, for each $i < n$ and $x\in \alpha_0$. 
For $S$ set,  let $\card{S}$ denote the number of elements of $S$.

%\subsection{Weight of formulas}
%\label{subsec:weights:sign}
The weight of a formula $A$ 
relative to $\model$ and $x$ is defined 
by taking into account the number of worlds  satisfying $A$ and the position of such worlds in $\sph{x}$. The weight is meant to indicate 
how implausible a formula is w.r.t. the actual world: 
namely, a formula $B$ satisfied by several worlds in a sphere is considered `highly plausible', and will have a low weight, while a formula $C$ satisfied by fewer worlds in the same sphere is considered `less plausible' than $B$  and will have  higher weight. 
Moreover, the definition takes into account 
the position of the worlds within the spheres, 
 with formulas satisfied by worlds in inner spheres being `more plausible' than those satisfied by worlds in outer spheres. 
Thus, the weight of a formula is defined by counting the worlds in each sphere satisfying the formula, starting from the innermost sphere.  
Formally:
\begin{definition}
\label{weight_of_formulas}
For any literal $A \in \Lcp$, the \emph{weight of $A$ relative to $\model $ and $x$}, denoted by $ \weight{x}{A} $,   
is the list of natural numbers 
$(w^A_0, .., w^A_n)$,  where 
$w^A_0 = \, \card{\{u \in \alpha_0 \mid u \modelsm{} A \}}$ and for $1 \leq i \leq n$:
\begin{itemize}[noitemsep]
    \item If $\alpha_i = \alpha_{i-1}$, then $w^A_i = w^A_{i-1}$;
    \item If $\alpha_i \supset \alpha_{i-1}$, then $w^A_i \, = \, \card{\{u \in \alpha_i \setdiff \alpha_{i-1} \mid u
 \models A \} } $. 
\end{itemize}
\end{definition}

%\todoavg{22/02 do we need If $\alpha_i = \alpha_{i-1}$?}
%\marianna{23 I think we do: otherwise, if $\alpha_i = \alpha_{i-1}$, then $w^A_i = 0$}

We next introduce a relation  $\xelleft x$ to compare literals according to their weights.  Literals smaller w.r.t. $\xelleft x$ have lower weight, and are thus the `more plausible' ones. The definition matches Lewis's intuition, according to which
worlds in  innermore spheres are `smaller' w.r.t. the similarity preorder than 
worlds in  outermore spheres.

\begin{definition} 
\label{def:weights:compare}
Let $\weight{x}{A}=(w^A_0, .., w^A_n)$ and  $\weight{x}{B}=(w^B_0, .., w^B_n)$, for some $A, B $ literals and $n \geq 0$. 
%To compare the weights, we 
We compare $\weight{x}{A}$ and $\weight{x}{B}$  {inverse lexicographically}, that is:  for the first $l\leq n$ such that $w^A_l \neq w^B_l$, set $\weight{x}{B} \presil \weight{x}{A} $ iff $w^A_l< w^B_l$. 
We say that \emph{$B$ has less or equal weight than $A$},  
%(equivalently \emph{$A$ has  higher or equal weight to $B$)  relative to $\model,x$}, 
in symbols $B\xelleft x A $,
iff $\weight{x}{B} \preil \weight{x}{A} $. 
We write $B \eq x A$ iff $B \xelleft x A$ and $ A \xelleft x B $. 
\end{definition}

%\todoavg{22/02 below phrasing is as if formulas can have many weights.}
%\marianna{23 Reformulated. Better?}
%\todoavg{23/02 yup, removed between, felt awkward.}
In finite models, 
the weight of a literal will be one of $2^{\card{W}}$ possible distinct values (lists). 
Literals satisfied by all worlds in $\bigcup \sph{x}$  have the highest weight (and are the most implausible ones), while literals falsified by all worlds have the lowest weight (and are the most plausible). It is easy to see that the relation $\xelleft{x}$ induces a well ordering among literals. 
Moreover, in centered models, $\weight{x}{A}\neq \weight{x}{\bar A}$, for any literal $A$. 
%(most plausible)  will be the weight of a formula satisfied by all worlds in $\bigcup \sph{x}$, and the highest (most implausible)  will be the weight of a formula satisfied by no world in $\bigcup \sph{x}$. Moreover, in finite models $\xelleft{x}$ induces a well ordering among formulas, so that it is always possible to compare formulas with respect to their weights.  
%In centered models, $\weight{x}{A}\neq \weight{x}{\bar A}$ for any formula $A$. 
%\todomar{can be shortened or removed} 
The following shows that the weights of a literal and its negation behave `dually'. The proof is in the Appendix. 

\begin{proposition}
\label{prop:Symmetry}
For  $A,B $ literals, if $B \xelleft x A$  then $\bar A \xelleft x\bar B$.
\end{proposition}

\begin{figure}[t]
	\begin{center}
		\begin{minipage}{0.26\textwidth}
		\hspace{-1.3cm}
        \small
			\begin{tabular}{ll}
				$p$& \text{``Nixon presses the button.''}  \\
				$e_1$
				& \text{``Technical error 1 occurs.''}\\
				$e_2$&  \text{``Technical error 2 occurs.''}\\
				$l$& \text{``The launch  is successful.''}\\
				$h$ & \text{``There is a nuclear} \\
                & ~~\text{holocaust.''}\\
			\end{tabular}
            \normalsize
		\end{minipage}
		\begin{minipage}{0.22\textwidth}
			\begin{tikzpicture}[scale=0.8]
				\draw[color=gray] (0, 0) ellipse (0.7cm and 0.7cm);
				\draw[color=gray] (0, 0.5) ellipse (0.9cm and 1.2cm);
				\draw[color=gray] (0, 1) ellipse (1.1cm and 1.7cm);
				\draw[color=gray] (0, 1.5) ellipse (1.3cm and 2.2cm);
				\draw[color=gray] (0, 2) ellipse (1.5cm and 2.7cm);
				
				\node[] at (0.1, -1.2) (x) {$\sph x$};
				\node[] at (0, 0) (x) {\tiny$\bullet$ \, \small $x$};
				\node[] at (0, 1.3) (vo) {\tiny$\bullet$ \, \small $z_1$};
				\node[] at (0, 0.9) (lvo) {\small$p, e_1, e_2$};
				\node[] at (0, 2.3) (vd) {\tiny$\bullet$ \, \small $z_2$};
				\node[] at (0, 1.9) (lvd) {\small$p, e_1, e_2$};
				\node[] at (0, 3.3) (yo) {\tiny$\bullet$ \, \small $y_1$};
				\node[] at (0, 2.9) (yol) {\small$p, l,h$};
				\node[] at (0, 4.3) (yd) {\tiny$\bullet$ \, \small $y_2$};
				\node[] at (0, 3.9) (ydl) {\small$p, l,h$};
			\end{tikzpicture}
		\end{minipage}
        \hspace{-0.2cm}\vline\hspace{0.1cm}
		\begin{minipage}{0.2\textwidth}
			\begin{tikzpicture}[scale=0.8]
				\draw[color=gray] (0, 0) ellipse (0.7cm and 0.7cm);
				\draw[color=gray] (0, 0.5) ellipse (0.9cm and 1.2cm);
				\draw[color=gray] (0, 1) ellipse (1.1cm and 1.7cm);
				\draw[color=gray] (0, 1.5) ellipse (1.3cm and 2.2cm);
				\draw[color=gray] (0, 2) ellipse (1.5cm and 2.7cm);
				
				\node[] at (0.1, -1.2) (x) {$\ups{\Gamma}{x}{\mathsf{d}}(x)$};
				\node[] at (0, 0) (x) {\tiny$\bullet$ \, \small $x$};
				\node[] at (0, 1.3) (vo) {\tiny$\bullet$ \, \small $y_1$};
				\node[] at (0, 0.9) (lvo) {\small$p, l,h$};
				\node[] at (0, 2.3) (vd) {\tiny$\bullet$ \, \small $y_2$};
				\node[] at (0, 1.9) (lvd) {\small$p, l,h$};
				\node[] at (0, 3.3) (yo) {\tiny$\bullet$ \, \small $z_1$};
				\node[] at (0, 2.9) (yol) {\small$p, e_1, e_2$};
				\node[] at (0, 4.3) (yd) {\tiny$\bullet$ \, \small $z_2$};
				\node[] at (0, 3.9) (ydl) {\small$p, e_1, e_2$};
			\end{tikzpicture}
		\end{minipage}
        \begin{minipage}{0.12\textwidth}
			\begin{tikzpicture}[scale=0.8]
				\draw[color=gray] (0, 0) ellipse (0.7cm and 0.7cm);
				\draw[color=gray] (0, 0.5) ellipse (0.9cm and 1.2cm);
				\draw[color=gray] (0, 1) ellipse (1.1cm and 1.7cm);
				\draw[color=gray] (0, 1.5) ellipse (1.3cm and 2.2cm);
				\draw[color=gray] (0, 2) ellipse (1.5cm and 2.7cm);
				
				\node[] at (0.1, -1.2) (x) {$\ups{\Sigma}{x}{\mathsf{d}}(x)$};
				\node[] at (0, 0) (x) {\tiny$\bullet$ \, \small $x$};
				\node[] at (0, 1.3) (vo) {\tiny$\bullet$ \, \small $z_1$};
				\node[] at (0, 0.9) (lvo) {\small$p, e_1,e_2$};
				\node[] at (0, 2.3) (vd) {\tiny$\bullet$ \, \small $z_2$};
				\node[] at (0, 1.9) (lvd) {\small$p, e_1,e_2$};
				\node[] at (0, 3.3) (yo) {\tiny$\bullet$ \, \small $y_1$};
				\node[] at (0, 2.9) (yol) {\small$p, l,h$};
				\node[] at (0, 4.3) (yd) {\tiny$\bullet$ \, \small $y_2$};
				\node[] at (0, 3.9) (ydl) {\small$p, l,h$};
			\end{tikzpicture}
		\end{minipage}
	\end{center}
\caption{
%A variation of the Nixon counterexample. 
\textbf{Left}: Meaning of atomic formulas
and system of spheres $\sph x$ within 
	$\model = \langle W, S, v\rangle$, for $W = \{x, z_1, z_2, y_1, y_2\}$, 
	%$\bigcup \sph{x} = W$, 
	$\sph{z_i}=\{\{z_i\}\}$ and $\sph{y_i}=\{\{y_i\}\}$ for $i \in\{1,2\}$, and 
	$
	\sph{x} = \{\{x\},\{x,z_1\},\{x,z_1,z_2\},\{x,z_1,z_2,y_1\},\{x,z_1,z_2,y_1,y_2\}\}
	$. 
	Moreover,
	$v(p) = \{z_1, z_2, y_1, y_2\}$; $v(e_1) = v(e_2) = \{z_1, z_2\}$; and $v(l)  = v(h) = \{y_1, y_2\}$.  
	\textbf{Right}: Updated systems of spheres 
	$\ups{\Gamma}{x}{\mathsf{d}}(x)$ and  $\ups{\Sigma}{x}{\mathsf{d}}(x)$.   Refer to Def.~\ref{def:update} and \cref{ex:running 2}.  
}
\label{fig:running 1}
\end{figure}

%The following example illustrates how weights are associated to formulas. 
\begin{example}\label{ex:running 1}
%\todoavg{22/02 if the example does more than just weights, put it after we defined updates?}
%\marianna{23 True, but I think it is good to have it here already, to see what weights do}
    Consider the centered  model $\model$ to the left of \cref{fig:running 1}, representing a variation of Nixon's counterexample. Propositions $e_1$ and $e_2$ represent two minor technical errors, which prevent the missile from launching. Since these are small changes w.r.t.\ the actual world when compared to a nuclear holocaust, worlds satisfying $e_1$ and $e_2$  are closer to $x$ than worlds satisfying $h$ in $\sph x$. 

    The two errors, even when considered together, represent a less significant change than the  miracle from Fine's example (\nc). 
    Considering \emph{two} very small changes 
    is a problem for approaches that evaluate \cpcounts by `counting' the number of formulas which worlds (fail to) satisfy. As we will argue in \cref{ex:running 2} and \cref{remark:critique_to_girard}, our approach can handle such cases. 
    %Figure~\ref{fig:running 1}  contains an interpretation of the propositional atoms, inspired by the \emph{Nixon counterexample} mentioned in the Introduction. 
    %Later, in \Cref{ex:running 2}, we will discuss the evaluation of specific \cpcount formulas. For the moment, observe that propositions $e_1$ and $e_2$ represent two minor technical problems, which prevent the missile from launching. 
    %Since these are small changes w.r.t.\ the actual world when compared to a nuclear holocaust, worlds satisfying $e_1$ and $e_2$  are closer to $x$ than worlds satisfying $h$.

    The weights of literals are calculated as follows:  
    $\weight{x}{p} = (0,1,1,1,1)$; 
    $\weight{x}{\bar p} =(1,0,0,0,0)$; 
    $\weight{x}{e_1} = \weight{x}{e_2} =    (0,1,1,0,0)$;
    $\weight{x}{l} = \weight{x}{h} = (0,0,0,1,1)$;
    $\weight{x}{\bar l} =(1,1,1,0,0)=\weight{x}{\bar h}$. 
    Comparing the weights according to Definition~\ref{def:weights:compare} yields:
	%\begin{equation}\label{eq:weights}
	$	\bar h  \eq x \bar l \xellefts x \bar e_1 \eq x \bar e_2 \xellefts x \bar p \xellefts x p \xellefts x e_1 \eq x e_2 \xellefts x h \eq x l.   $
%	\end{equation}  
 %    $$
	% \begin{array}{c  c  l}
	% 	\weight{x}{p} & = &(0,1,1,1,1)\\[0.05cm]
	% 	\weight{x}{\bar p} &=&(1,0,0,0,0)\\[0.05cm]
	% 	\weight{x}{e_1} & =  &  (0,1,1,0,0) = \weight{x}{e_2}\\[0.05cm]
	% 	\weight{x}{l} & = & (0,0,0,1,1) =\weight{x}{h}\\[0.05cm]
	% 	\weight{x}{\bar l} &=&(1,1,1,0,0)=\weight{x}{\bar h}
	% \end{array}
	% $$
\end{example}

%Next, we show that our definition of weights acts according to the following intuition. 
%If a formula $A$ has high weight, then it should be considered as highly implausible. 
%Naturally then, 
%we expect the dual $\bar A$ 
%to have low implausibility  (or  high plausibility). 

%\subsection{Weight of sets of formulas}
%\label{subsect:weight:sets}
%In the following, we want to prioritise more significant  over less significant formulas. This  will also be necessary when comparing  sets of formulas.
%Intuitively, given a (finite) set of formulas $\Gamma=\{A_j\}_{j}$, its weight  is  set to be a list of weights of formulas defined in the following way:  first the formulas in $\Gamma$ are ordered w.r.t. significance (high to low and left to right) and then  the weight $\weight{x}{A_j}$ of  each $A_j\in \Gamma$ is calculated. Alternatively, due to \cref{prop:Symmetry}, this amounts to simply calculating all the weights $\weight{x}{A_j}$ and writing them from higher to lower  (left to right). Then, to compare the weights of two sets of formulas we lexicographically 
%compare the lists w.r.t. the weight of their elements. We will now  formally present the second approach:
Next, we define weights of sets of literals. This amounts to calculating the weights of all the literals in the set, and ordering them from higher to lower  (left to right). The resulting lists are compared lexicographically. Formally: 
\begin{definition}
\label{weight of sets}
For a set of literals $\Gamma = \{ A_1, .., A_m \}$, the \emph{weight of $\Gamma$ relative to $\model$ and $x$}, written $\weight{x}{\Gamma}$, is the list of weights of literals $(w_1, .., w_m)$ where, for $i< m$ and $j, k \leq m$, we set $w_i = \weight{x}{A_j}$ and  $w_{i+1} = \weight{x}{A_k}$ if and only if  $ A_k \xelleft x A_j$.
We set  $\weight{x}{\{\varnothing\}}=(0)$.
\end{definition}

\begin{definition}
\label{def:weight:sets:compare}
For $\Gamma$, $\Delta$ sets of literals, let $
\weight{x}{\Gamma}=\{c_1,..,c_m\}$ and $              \weight{x}{\Delta}=\{d_1,.. , d_k\}$, for some $m, k \geq 0$. 
%To compare $\weight{x}{\Gamma}$ and $\weight{x}{\Delta}$, 
We compare  $\weight{x}{\Gamma}$ and $\weight{x}{\Delta}$ lexicographically, that is: 
for the first $l \leq \min\{m,k\}$ such that $c_l\neq d_l$, set $\weight{x}{\Gamma} \presl  \weight{x}{\Delta} $ if and only if $   c_l \xelleft x d_l $. If, for all $l\leq min\{m,k\}$ it holds that $d_l = c_l$, we set $\weight{x}{\Gamma} \prel  \weight{x}{\Delta} $ if and only if $m\leq k$. 
We say that \emph{$\Gamma$  has less or equal weight than $\Delta$
%} and write , 
%(equivalently \emph{$\Delta$ has higher or equal weight to $\Gamma$) 
relative to $\model, x$}, in symbols $\Gamma \lexleft x \Delta$, iff $\weight{x}{\Gamma} \prel  \weight{x}{\Delta}$. 
We write 
$\Gamma \eq x \Delta $ whenever $\Gamma \lexleft x \Delta$ and $ \Delta \lexleft x \Gamma $.  
\end{definition}

% \begin{remark}
%     When comparing $\weight{x}{\Gamma}$ and $\weight{x}{\Delta}$
%  it is important that we prioritise  formulas w.r.t. their significance. By \cref{def: significance}, this is done by comparing their weight (intuitively their implausibility). 
%  To obtain this, it is essential to order 
%  the weights from higher to lower (more implausible to less). 
% \footnote{ However, note that 
%  listing the weights from left to right, as done in  \cref{def:weight:sets:compare}, is just a convention. Doing it from right to left instead, would simply mean we should use the colexicographic order, i.e., same as lexicographic but from right to left.}
% \end{remark}
%\todomar{killed a remark here}

The next result, whose proof is in the Appendix, shows that the subset relation is monotone over weights. 

\begin{proposition}%[Monotonicity of weight]
\label{monotonicity}
For $\Gamma, \Delta$ sets of literals with $\Gamma \subseteq \Delta$, it holds that $\Gamma \lexleft x \Delta$. 
\end{proposition}

%\todomar{Tried to write what we discussed but failed. I don't remember. Please help! The old part is commented.}
%\todoavg{22/02 What you wrote was correct, i changed it a bit. Becomes easier by invoking \cref{prop:Symmetry}. Yours lies commented}
%\marianna{23 I reformulated a bit and added a part at the end. I think it works!}

When evaluating \emph{cp}-counterfactuals, we will calculate the weights of specific sets of literals - namely, the literals in a \emph{cp}-set that are different between the actual world and other possible worlds. \cref{prop:Symmetry} naturally induces a spectrum of weights of literals, from low (hight implausibility) to high (low implausibility). 
If a literal $A$ is placed at the low end of the spectrum, its negation $\bar A$ is dually placed at the high end.  Those pairs (literals and their negations) which find themselves closer to the ends of this spectrum are precisely the ones which will be crucial for the evaluation. To evaluate a \cpcount, we will `re-arrange' the worlds in a system of spheres, according to the (im)plausibility of the literals in the \emph{cp}-set that the worlds (fail to) satisfy. Thus, the `significant' pairs of literals for the evaluation will be those that have very low and very high implausibility, as worlds satisfying or failing to satisfy them will `move' very close to or very far from the actual world in the updated system of spheres.  

%and we identify them as the `significant' ones. Then, we shall re-order the system of spheres accordingly: worlds satisfying formulas of high weight (high implausibility) will be placed further away from the center of the system of spheres and vice versa.  

%\todoavg{22/02 changed the last few lines}
%\marianna{23 changed a bit and added the last part}
Formally,  formula $A$ is \emph{at least as significant} as  formula $B$ iff $\{A, \bar A\} \lexleft x \{B, \bar B\}$.
%In this sense, the notion of weight of formulas allows us to identify the formulas which will be important in the evaluation of a \emph{cp}-counterfactual. 
%In \cref{ex:running 1}, we have that $\max \{\weight{x}{e_1}, \weight{x}{\bar e_1}\} = e_1 \xellefts x h = \max \{\weight{x}{h}, \weight{x}{\bar h}\} $ or equivalently,  $ \min \{\weight{x}{h}, \weight{x}{\bar h}\}= \bar h \xellefts x \bar e_1 = \min \{\weight{x}{e_1}, \weight{x}{\bar e_1}\}  $ from which we obtain that  $\bar h,h$ are more significant than $\bar e_1,e_1$. 
In \cref{ex:running 1}, we have
$\{e_1, \bar{e_1}\} \lexleft x \{h, \bar{h}\}$, as $\max \{\weight{x}{e_1}, \weight{x}{\bar e}_1\} = e_1 \xellefts x h = \max \{\weight{x}{h}, \weight{x}{\bar h}\} $. Similarly, $\{e_2, \bar{e}_2\} \lexleft x \{h, \bar{h}\}$, whence $h$ is more significant than $e_1$ and $e_2$. 
Also, since $\weight{x}{l} = \weight{x}{h}$,
%we have $\{l, \bar l\} \lexleft x  \{e_1, \bar{e}_1\}$ and $\{l, \bar l\} \lexleft x  \{e_2, \bar{e}_2\}$. 
we conclude that $l$ is more significant than $e_1$ and $ e_2$.

\section{Evaluating \emph{cp}-counterfactuals}
\label{sec:update}
%\todomar{changed title}
In this section we discuss how to evaluate a \emph{cp}-counterfactual $A \discond{\Gamma} B$ at a world $x$ of a (weakly) centered sphere model $\model$. We shall first `re-arrange', or `update', the system of spheres $\sph x$ by calculating, for each world $y$, the weight of the \emph{disagreement set} between $x$ and $y$, that is, the set of literals in $\Gamma$ which are satisfied at $x$ but not at $y$. 
Intuitively, the worlds whose disagreement sets have low weights   are the `most plausible', as they differ less from the actual world $x$. 
We define the updated system of spheres $\ups{\Gamma}{x}{}(x)$  by accordingly placing such worlds in the innermost spheres. 
The Lewis counterfactual $A \discond{} B$ is then evaluated at $\ups{\Gamma}{x}{}(x)$. We start by defining disagreement sets and the corresponding disagreement update.  

%\todoavg{23/02 slightly edited above}

\begin{definition}
\label{def:prioritisations}
For $\model = \langle W, S, v\rangle$,  $x,y \in W$ and $\Gamma $ \emph{cp}-set of formulas, the \emph{disagreement set of $y$ w.r.t. $\Gamma$ and $x$} is  $\disset{x}{y}{\Gamma} = \{G \in \Gamma \mid x \modelsm{} G \textit{ iff } y \not \modelsm{} G \}$. 
\end{definition}

\begin{definition}
\label{def:update}
Given a (weakly) centered model $\mathcal{M}= \langle W, S, v\rangle$, $x \in W$ and a \emph{cp}-set $\Gamma$, the  \emph{disagreement update of $\model$ at $x$ and $\Gamma$}, denoted  $\upmod{\Gamma}{x}{\mathsf{d}}= \langle \upw{\Gamma}{x}{d}, \ups{\Gamma}{x}{\mathsf{d}}, \upv{\Gamma}{x}{\mathsf{d}}\rangle$, is defined  by setting $\upw{\Gamma}{x}{\mathsf{d}} = W$ and $\upv{\Gamma}{x}{\mathsf{d}} = v$. To define $\ups{\Gamma}{x}{\mathsf{d}}$, we first define the spheres of $\upsph{d}{x}$. For $n\in \{\weight{x}{\disset{x}{y}{\Gamma}} \mid {y} \in \bigcup \sph x\}$, let  
%and for $i < \card{\bigcup \sph{x}} \, $, define: 
%\todomar{removed the d from notation, changed the notation (it is a macro!)}
%\todoavg{22/02 no d is fine but why the brackets []? It seems unecessary}
%\marianna{23 removed parenthesis. Is it readable?}
$
%\label{eq:spheres}
\sigma_n = \{ {y} \in 
\bigcup \sph{x}
\mid \weight{x}{\disset{x}{y}{\Gamma}} \prel n\}.     
$
Then, for each $n$ and for each $\card{\bigcup \sph{x}} \, > i \geq 0$, we inductively define `subspheres' $\sigma_{n_i}$ of $\sigma_n$: 
\begin{equation*}
\begin{split}
%	\label{up1}
&\sigma_{n_0} \! =   \sigma_{n-1} \cup \{y \in
 (\sigma_n \!\setdiff \sigma_{n-1}) 
 \!
\mid \!\!
\text{ for all }  z  
\in (\sigma_n\! \setdiff \sigma_{n-1}),   \alpha \in \sph{x}, \! \text{ if } z \in \alpha \text{ then } y \in \alpha \}  
\\
%\label{up2}
&\sigma_{n_{i+1}} \! = \sigma_{n_i} \cup \{y \in
 (\sigma_n \setdiff\sigma_{n_i}) 
\mid \!\!
\text{ for all }  z 
\in (\sigma_n \setdiff \sigma_{n_i}),  \alpha \in \sph{x}, \!
\text{ if } z \in \alpha \text{ then } y \in \alpha \} 
\end{split}
\end{equation*}
Finally, set $\ups{\Gamma}{x}{\mathsf{d}}(x) = \{\sigma_{n_i}\}_{n,i} $ and for all $y \neq x$, $\ups{\Gamma}{x}{\mathsf{d}}(y) = \sph{y}$. 
\end{definition}

From the definition it follows that, for every $i$, $\sigma_{n-1} \subseteq \sigma_{n_i} \subseteq \sigma_n$. Moreover, since $\model$ is finite, $\sigma_{n_j}=\sigma_n$, for some   $j\leq \,\card{\bigcup \sph{x}}  $.   
Intuitively, for any $n$ weight of some disagreement sets of worlds in $\sph x$, 
the sphere $\sigma_n$ contains exactly all worlds whose disagreement set w.r.t $\Gamma$ and $x$ has weight at most $n$ (according to $\lexleft x$ or, equivalently, to $\prel$). 
The subspheres $\sigma_{n_i}$ are then needed to `distinguish' between worlds $y_1$, $y_2$ whose disagreement sets have the same wights, but which are placed within different spheres in $\sph x$. In this case, we want $y_1$ and $y_2$ to be placed in distinct spheres in $\ups{\Gamma}{x}{\mathsf{d}}(x)$, whence the introduction of the subspheres. Next, we show that the models obtained through the updates are sphere models (the proof is in the Appendix). 

\begin{lemma}
\label{update properties}
Given $\model$, $x$ and  $\Gamma$,  
$\upmod{\Gamma}{x}{\mathsf{d}}$ satisfies non-emptiness and nesting. Moreover, if $\model$ is (weakly) centered then $\upmod{\Gamma}{x}{\mathsf{d}}$ is (weakly) centered.
%, and if $\model$ is centered then $\upmod{\Gamma}{x}{\mathsf{d}}$ is  centered.
\end{lemma}

We now define satisfaction of formulas. We use the following shorthands: 
$\alpha \fem{\mathsf{d}} A$ \emph{iff} $ \text{there is } y \in \alpha \text{ s.t. } y \modelsm{\mathsf{d}}  A$, 
and $\alpha \fum{\mathsf{d}} A $ \emph{iff} $  \text{for all } y \in \alpha,  y \modelsm{\mathsf{d}} A$.

\begin{definition}
\label{def:satisfaction}
The satisfaction relation of formulas $A \in \Lcp$ at a world $x$ of a model $\model$ is defined by adding to the clauses in Definition~\ref{def:sat_atomic} the following:
%    For a sphere model $\model$  and $x$ world in it, we define the following 
  %  Then, we inductively define satisfaction of formulas as follows:
    \begin{itemize}[noitemsep]
    \item $\model, x \not\modelsm{\mathsf{d}} \bot $;
   % \item $\model, x \modelsm{\mathsf{d}} p$ \emph{iff} $ x \in  v(p) $;
    \item $\model, x \modelsm{\mathsf{d}} A \rightarrow B $ \emph{iff} $\model, x \not\modelsm{\mathsf{d}} A $ or $ \model, x \modelsm{\mathsf{d}} B $;
    \item $	\model, x \modelsm{\mathsf{d}} A \discond{} B $ \emph{iff}  there is $  \alpha \in \sph{x}$ such that $ \alpha \fem{\mathsf{d}}A$,  then there is $ \beta \in \sph{x}$ such that $\beta \fem{\mathsf{d}} A $ and $ \beta \fum{\mathsf{d}} A \rightarrow B$; 
    \item  $\model, x \modelsm{\mathsf{d}} A \discond{\Gamma} B$ \emph{iff} $ \upmod{\Gamma}{x}{\mathsf{d}}, x \modelsm{\mathsf{d}} A \discond{} B$.
    \end{itemize}
   We say that $A$ is \emph{valid in weakly centered} (resp. \emph{centered}) \emph{models models under the disagreement update}, in symbols 
   %$\validm{\mathsf d}^{\mathsf{VW}} A$ 
   $\validm{}^{\mathsf{VW}^\mathsf{d}} A$
   (resp. $\validm{}^{\mathsf{VC}^{\mathsf d}} A$), iff  $\model, x \modelsm{\mathsf{d}} A$ holds for for all worlds $x$ and all weakly centered (resp. centered) models $\model$. 
\end{definition} 

The clause for $\cond$ above corresponds to the truth condition for Lewis counterfactual. Thus, when restricted to formulas in $\Ll$, validity in  (weakly) centered sphere models under the disagreement update coincides with validity in  (weakly) centered sphere models. 
Formally, for $A \in \Ll$, we say that $A$ is \emph{valid in weakly centered} (resp. \emph{centered}) \emph{sphere models}, denoted by $\validm{}^{\mathsf{VW}} A$ (resp. $\validm{}^{\mathsf{VC}} A$) \emph{iff}  $\model, x \modelsm{\mathsf{d}} A$ holds for for all worlds $x$ and all weakly centered (resp. centered) models $\model$. These sets of validities identify Lewis' logics $\VW$ and $\VC$ respectively. 

The sets of formulas valid under the disagreement update in turn give rise to two logics, which we denote by $\VWu{d}$ and  $\VCu{d}$. 
We have just observed that, for $A \in \Ll$, $\validm{}^{\mathsf{VW}} A$ \emph{iff} $\validm{}^{\mathsf{VW}^\mathsf{d}} A$ and $\validm{}^{\mathsf{VC}} A$ \emph{iff} $\validm{}^{\mathsf{VC}^\mathsf{d}} A$. 
In the next section, we shall prove that similar statements hold for formulas of $ \Lcp$, thus establishing soundness and completeness of $\VWu{d}$ and  $\VCu{d}$ w.r.t. Lewis' logics $\VW$ and $\VC$ respectively.
We conclude this section by illustrating with an example the satisfiability under the disagreement update, and comparing our updates with the prioritarisations from the literature. 

\begin{example}\label{ex:running 2}
Recall the model $\model$ from 
\cref{ex:running 1}. 
Since  worlds satisfying $e_1$ and $e_2$  are closer to $x$ than worlds satisfying $h$, the counterfactual from Nixon's counterxample $p\discond {}h$ is evaluated as false at $\sph x$, that is: $\model, x \not \models p\discond {}h$.  
Let us now consider  the \emph{cp}-sets $\Gamma = \{e_1,e_2,\bar e_1,\bar e_2\}$, $\Sigma=\Gamma \cup \{l,\bar l\}$ and the \cpcounts $p \discond\Gamma h$ and $p \discond{\Sigma} h$. 
%\marianna{changed $\Gamma'$ to $\Sigma$}
They respectively express ``If Nixon had pressed the  button, and no technical errors occurred, there would have been a nuclear holocaust'' and ``If Nixon had pressed the button,  no technical errors occurred, but the launch was not successful, there would have been a nuclear holocaust.''
Intuitively, the first sentence should be evaluated as true, and the second as false. We  evaluate  the  \cpcounts in our framework.

 %We have already  obtained the weights of all literals and compared them in \cref{ex:running 1}:  $\bar h  \eq x \bar l \xellefts x \bar e_1 \eq x \bar e_2 \xellefts x \bar p \xellefts x p \xellefts x e_1 \eq x e_2 \xellefts x h \eq x l   $. 
 Continuing from \cref{ex:running 1}, we need to consider the disagreement sets of each world w.r.t. $x$, and compare their weights (Definition~\ref{def:weight:sets:compare}), considering $\Gamma $ and $\Sigma$ respectively. The disagreements sets are: 
 $\disset{x}{x}{\Gamma}  = \disset{x}{y_1}{\Gamma}=\disset{x}{y_2}{\Gamma} =  \varnothing $; 
 $\disset{x}{z_1}{\Gamma} =\disset{x}{z_2}{\Gamma}  =   \Gamma$
 %\{e_1, e_2, \bar e_1, \bar e_2\}  $;
 and 
$\disset{x}{y_1}{\Sigma} =\disset{x}{y_2}{\Sigma}  =  \{l,\bar l\}$; 
and $\disset{x}{z_1}{\Sigma} =\disset{x}{z_2}{\Sigma}  = \Gamma$. 
 %$$
%\begin{array}{r c c c l}
%  \disset{x}{x}{\Gamma} & =& \disset{x}{y_1}{\Gamma}=\disset{x}{y_2}{\Gamma} & = & \varnothing  \\[0.1cm]
%  & &\disset{x}{z_1}{\Gamma} =\disset{x}{z_2}{\Gamma} & = &  \{e_1, e_2, \bar e_1, \bar e_2\}   \\[0.1cm]
%  & &\disset{x}{y_1}{\Sigma} =\disset{x}{y_2}{\Sigma} & = & \{l,\bar l\}   \\[0.1cm] & &\disset{x}{z_1}{\Sigma} =\disset{x}{z_2}{\Sigma}  & = & \{e_1, e_2, \bar e_1, \bar e_2\} 
%  \end{array}
%$$
%\marianna{if we have space, we could write the weights of disagreement sets which we use for the comparison}
Comparison of their weights yields:
$ \varnothing \eq{x}\disset{x}{x}{\Gamma} \eq{x}\disset{x}{y_1}{\Gamma} \eq{x} \disset{x}{y_2}{\Gamma} \lexlefts x \disset{x}{z_1}{\Gamma} \eq{x}\disset{x}{z_2}{\Gamma} \eq{x}\Gamma $
 and 
$\Gamma \eq{x} \disset{x}{z_1}{\Sigma} \eq{x} \disset{x}{z_2}{\Sigma}  \lexlefts x \disset{x}{y_1}{\Sigma} \eq{x}\disset{x}{y_2}{\Sigma} \eq{x}\{l,\bar l\}   $.
The two updated systems of spheres $\ups{\Gamma}{x}{\mathsf{d}}(x)$ and $\ups{\Sigma}{x}{\mathsf{d}}(x)$ are displayed on the right of \cref{fig:running 1}. 
%Finally, we calculate the updated systems of spheres following Definition~\ref{def:update}  (refer to \cref{fig:example}): 
%$$
%\ups{\Gamma}{x}{\mathsf{d}}(x) =\{\{x\}, \{x,y_1\},\{x,y_1,y_2\},\{x,y_1,y_2,z_1\}, \{x,y_1,y_2,z_1,z_2\}\} 
%$$ 
%$$
%\ups{\Sigma}{x}{\mathsf{d}}(x) =\{\{x\},\{x,z_1\},\{x,z_1,z_2\},  \{x,z_1,z_2,y_1\}, \{x,z_1,z_2,y_1,y_2\}\}
%$$
It holds that $\upmod{\Gamma}{x}{\mathsf d}, x \modelsm{\mathsf d} p \discond{} h$, and $\upmod{\Sigma}{x}{\mathsf d}, x \not\modelsm{\mathsf d} p \discond{} h$. 
The disagreement update yields the intuitively correct result on both formulas.

Observe that, on $p \discond{\Sigma} h$, the updated system of spheres $\ups{\Sigma}{x}{\mathsf{d}}(x)$ `ranks' the worlds according to the significance of the literals in $\Sigma$ they satisfy. Thus, worlds $y_1, y_2$ are placed in the outermost spheres, as they `disagree' with $x$ on $l$, the most significant formula in $\Sigma$.  %satisfy formula $l$ which is more significant than $e_1$ and $e_2$, and it is not satisfied at $x$. 
%\marianna{24 maybe a short sentence here recalling significance?}
%\todoavg{24 fixed example, typos and math issues. Included third ellipse, made em fit}
\end{example}

%\begin{remark}
%\label{remark:paired}
%\todomar{this will go to the appendix or it will be killed}
%Let us briefly showcase the reason we require the \emph{ceteris paribus} sets to be paired. Consider the  sets $\Delta=\{e_1,e_2,l\}$, $\Delta'=\{\bar e_1,\bar e_2,\bar l\}$ and let us try to evaluate the \cpcounts $p\discond \Delta h$, $p\discond {\Delta'} h$. Intuitively, requiring that ``all things in $\Delta$ are unchanged'' is logically equivalent to ``all things in $\Delta'$ are unchanged''.  However, $\{e_1, e_2 \} \eq{x} \disset{x}{z_1}{\Delta} \eq{x} \disset{x}{z_2}{\Delta}  \lexlefts x \disset{x}{y_1}{\Delta} \eq{x}\disset{x}{y_2}{\Delta} \eq{x}\{l\}   $ and $\{\bar l\}\eq{x}   \disset{x}{y_1}{\Delta'} \eq{x}\disset{x}{y_2}{\Delta'}     \lexlefts x  \{\bar e_1, \bar e_2 \} \eq{x} \disset{x}{z_1}{\Delta'} \eq{x} \disset{x}{z_2}{\Delta'} $. It holds that $\upmod{\Delta}{x}{\mathsf d}, x \not\modelsm{\mathsf d} p \discond{} h$, and $\upmod{\Delta'}{x}{\mathsf d}, x \modelsm{\mathsf d} p \discond{} h$. Without pairing we may  get contradictory results! A similar situation arises if we consider the Nixon example from \cite{DBLP:journals/corr/GirardT16,girard2018prioritised} for $\Delta=\{m,l\},\Delta'=\{\bar m,\bar l\}$ and evaluate the counterfactuals $p\discond \Delta h$, $p\discond {\Delta'} h$ w.r.t. our updates. 
%
%%This happens because we differentiate formulas based on their significance, \cref{def: significance}.
%
%
%\end{remark}
%\marianna{commented remark about pairedness here}

\begin{figure}[t!]  \vspace{28pt} 
	\centering
	\begin{tabular}{| @{\hspace{0.2cm}} c %@{\hspace{0.1cm}} | @{\hspace{0.1cm}}  c 
    @{\hspace{0.2cm}} || @{\hspace{0.2cm}}  c @{\hspace{0.2cm}} | @{\hspace{0.2cm}}  c @{\hspace{0.2cm}} | @{\hspace{0.2cm}}  c @{\hspace{0.2cm}} || @{\hspace{0.2cm}} c @{\hspace{0.2cm}} | }
	%	\hline
	%	\multicolumn{6}{|c|}{Comparison of \emph{ceteris paribus} evaluations} \\
    \hline
    \emph{cp}- 
    %& $cp$-set 
     & strict & na\"ive & maximal &disagreement\\
     counterfactual&   evaluation & counting & superset &update \\
    \hline

			 & & & & \\[-0.3cm]
		$p\discond{\Gamma} h$  
        %& $\Gamma$
        & \checkmark&  \checkmark & \checkmark&\checkmark\\
		$p\discond{\Sigma} h$
        %&    $\Sigma$  
        & \checkmark&\checkmark & $\times$ &$\times$\\
		%\hline
		$p\discond{\Gamma} \bar h$ 
        %&  $\Gamma$ 
        &   $\times$  & $\times$ &$\times$ &$\times$\\
		$p\discond{\Sigma} \bar h$
        %&   $\Sigma$ 
        &\checkmark&$\times$ &$\times$ &\checkmark\\
		\hline
		
	\end{tabular}
	\vspace{0.3cm}
	
\caption{Results of evaluating \cpcounts at world $x$ of the model from \cref{ex:running 1}. Symbol $\checkmark$ means that the formula is satisfiable at $x$, and $\times$ that it is not. 
Strict evaluation, na\"ive counting and maximal superset are from~\cite{DBLP:journals/corr/GirardT16,girard2018prioritised} (strict evaluation is there called \emph{ceteris paribus}), while disagreement update is our proposed evaluation. 
%CP stands for \emph{ceteris paribus}, and it corresponds to the first simple evaluation from \cite{DBLP:journals/corr/GirardT16,girard2018prioritised} (we briefly discussed it in the introduction under the name of `strict evaluation'). Then NC and MS stand for \emph{na\"ive counting} and \emph{maximal supersets}, also from \cite{DBLP:journals/corr/GirardT16,girard2018prioritised}, discussed in this section, while DU is  our disagreement update.
}
\label{fig: analysis of evaluations}
\end{figure}

\begin{remark}[Comparison with \cite{girard2018prioritised,DBLP:journals/corr/GirardT16}]
\label{remark:critique_to_girard}
We now consider  the  prioritisations \emph{na\"ive counting} and \emph{maximal supersets} defined by Girard and Triplett in~\cite{girard2018prioritised,DBLP:journals/corr/GirardT16}. 
%\marianna{mention somewhere (introduction?) the strict evaluation}
%\marianna{mention somewhere that Girand and co cannot do nested operators}
These prioritarizations operate on preferential centered models, and detail how to `re-arrange' the preorder relation to evaluate formulas. Both notions rely on  \emph{agreement sets}, defined in \cite{DBLP:journals/corr/GirardT16} as $A^\mathcal{M}_\Gamma(x, y)=\{G\in \Gamma \mid x\modelsm{} G \text{ iff } y\modelsm{} G\}$. 
%\marianna{killed footnote about colors}
%\footnote{In \cite{girard2018prioritised}, the authors introduce an equivalent notion of agreement, by first defining `colors' of worlds ${Col}^\mathcal{M}_\Gamma(y)=\{G\in \Gamma \mid y\models G\}\cup \{\bar G \mid G\in \Gamma \text{ and }y \not\models \bar G\}$, and `palettes' of sets, $Pal(\Gamma)=\{G,\bar G\mid G\in \Gamma\}$, which are akin to our requirement of paired \emph{cp}-counterfactuals. Then,  $A^\mathcal{M}_\Gamma(y,v)={Col}^\mathcal{M}_\Gamma(y)\cap {Col}^\mathcal{M}_\Gamma(v)$.}.
Na\"ive counting simply counts the number of formulas two worlds $z$, $x$ agree on.
%, and no formula is treated as more significant than any other. 
Thus, a world $y$ is `closer' to $x$ than $z$ if  $|A^\mathcal{M}_\Gamma(x,y)|\geq |A^\mathcal{M}_\Gamma(x,z)|$. 
Under maximal supersets, a world $y$ is `closer' to $x$ than $z$ if $A^\mathcal{M}_\Gamma(x,z)\subseteq A^\mathcal{M}_\Gamma(x,y)$. 
%If $A^\mathcal{M}_\Gamma(x,v)\cap A^\mathcal{M}_\Gamma(x,y) $ is neither $A^\mathcal{M}_\Gamma(x,v)$ nor $A^\mathcal{M}_\Gamma(x,y)$, the worlds $v$ and $y$ are incomparable. 
Our disagreement sets 
(Definition~\ref{def:prioritisations}) are %defined dually to \cite{DBLP:journals/corr/GirardT16},  
dual to the agreement sets, and the definitions in~\cite{girard2018prioritised,DBLP:journals/corr/GirardT16} can be easily adapted to sphere models. 
Unlike na\"ive counting and maximal supersets, our disagreement update differentiates formulas in \emph{cp}-sets by taking into account their significance or, equivalently, by considering their weights.  
Instead, the prioritarisations of Girard and Triplett are based on `counting' the number of formulas in specific sets. This can lead to counterintuitive evaluations of quite simple formulas. In \cref{fig: analysis of evaluations} we report the results of evaluating formulas \emph{cp}-formulas in the centered model $\model$ from \cref{ex:running 1}. While maximal supersets gives the intuitively correct evaluation of the formulas considered in \cref{ex:running 2}, none of the evaluations from~\cite{DBLP:journals/corr/GirardT16} yields the intuitively correct result over formulas $p\discond{\Gamma}\bar h$ and $p\discond{\Sigma}\bar h$, which our disagreement update evaluates correctly as satisfiable and not satisfiable at $\model, x$. 

Moreover, the \cpcount under the strict evaluation of~\cite{DBLP:journals/corr/GirardT16} is not  \emph{dynamic} (in the sense of Dynamic Epistemic Logic, refer, e.g., to~\cite{van2007dynamic}). The strict evaluation performs a `one-step' update of the model. If the formula contains several \emph{cp}-modalities, these are not be taken into account. 
The authors do not discuss whether naïve counting and maximal supersets allow for iterated updates. 
Our disagreement update allows to evaluate formulas with nested \emph{cp}-modalities.
\end{remark}

\section{Soundness and completeness}
\label{sec:completeness}

\newcommand{\VL}{\mathsf{VL}}
\newcommand{\VLu}[1]{\mathsf{VL}^{\mathsf{#1}}}

We now turn to proving soundness and completeness of $\VWu{d}$ and $\VCu{d}$ with respect to the  Lewis' logics $\VW$ and $\VC$ respectively. For this section, set $\mathsf{L}\in\{\mathsf{W}, \mathsf{C}\}$. We shall write $\VL$ (resp. $\VLu d$) to denote $\VW$ or $\VC$ (resp. $\VWu{d}$ and $\VCu{d}$). 

Let us first illustrate our proof strategy. 
For convenience, we shall consider a language $\Lcpp$ having as primitive a different \emph{cp}-operator: $\disless{\Gamma}$, which is the \emph{ceteris paribus} version of the \emph{comparative plausibility} operator $\less$,  introduced by Lewis in~\cite{lewis1973}. 
We denote by $\Llp$ the language featuring only the $\less$ operator (without \emph{cp}-sets).  
The comparative plausibility and the conditional operator are interdefinable, and it is easy to show that the same holds for their \emph{cp}-versions. 
Then, we shall prove that  for any $A \in \Lcpp $,  $ \validm{}^{\VL} A$ iff  $ \validm{}^{\VL^\mathsf{d}} \hat A$, where $\hat A\in \Llp$ is a formula equisatisfiable with $A$ but having empty \emph{cp}-sets. 
%and that $b)$ if $ \validm{\mathsf{d}}^{\VL} A$, then there is a formula $A^\less\in \Llp$ such that $\validm{}^{\VL} A^\less$.  
%Proving that if  $ \validm{}^{\VL} A$ then  $ \validm{\mathsf{d}}^{\VL} A$ is immediate. However, the converse direction does not hold, as 
%not every \emph{cp}-comparative plausibility formula is a 
%However, while $a)$ is immediate, statement $b)$ does not hold: it is not 
%holds is far less trivial, as we need to associate to 
%each $A \in \Lcpp$ an equisatisfiable formula in $A^\less\in \Llp$,  having empty \emph{cp}-sets. 
% having only the comparative plausibility $\disless{\Gamma}$ as primitive modal operator. Since the comparative plausibility and the counterfactual are interdefinable, to each valid formula in $\Lcp $ there correspond a valid formula in $\Lcpp$, and vice-versa. 
% \marianna{say something about the comparative plausibility operator}
%Soundness is immediate, because Lewis' conditional operators are just \emph{cp}-conditionals with empty \emph{cp}-sets. However, completeness is far less trivial, as we need to associate to each \emph{cp}-comparative plausibility a formula within Lewis' logic (having only empty \emph{cp}-sets). 
%Thus, we shall consider a language $\Llp$ having only $\less$ as modal operator, and show that to every formula in $\Lcpp$ there corresponds an equisatisfiable formula in $\Llp$. 
Constructing formula $\hat A$ is far from trivial, but it allows us to relate the \cpcount evaluated in the disagreement update with Lewis' counterfactuals. 
The  construction is based on ideas from~\cite{girard2018prioritised}, adapted to the disagreement update. 
Soundness and completeness of $\VLu d$ w.r.t. $\VL$ then follow immediately. 
%The proof structure is similar to the one from \cite{girard2018prioritised}, but the step from $\Lcpp$ to $\Llp$ requires some non-trivial adjustments. 
We choose the (\emph{cp}-)comparative plausibility as our primitive operators because their truth conditions are simpler than the ones for the (\emph{cp}-)counterfactual, and thus simplify the construction. 
%\marianna{mention that the soundness and completeness results are not unexpected, but finding the formula is far from trivial}

We start by defining the languages $\Lcpp$ and $\Llp$. The formulas of $\Lcpp$ are generated from a countable set of propositional atoms $\atm $, by means of the grammar  
$
A ::= p \mid
\bot 
\mid A \rightarrow A
\mid A\disless{\Gamma}A
$, 
where $p \in  \atm$ and $\Gamma $ \emph{cp}-set of literals of $\Lcpp$ (so $\Gamma$ is finite and paired). As before, we write $A \less B$ for $A \disless{\varnothing} B$. 
We denote by $\Llp$ the language generated by the above grammar with the restriction that $\Gamma = \{\varnothing\}$. 
The operator $A \less B$ is read as ``$A$ is at least as plausible as $B$'', and  $A\disless{\Gamma}B$ is its \emph{cp}-version: `All things in $\Gamma$ being equal, $A$ is at least as plausible as $B$'. 
Moreover, following Lewis, we define  $\Diamond A := \lnot (\bot \less A)$. %This shorthand will be useful in the statement and proof of \Cref{ap lemma:paired:dis} below. 

Given a (weakly) centered model $\model$, a world $x$ and a \emph{cp}-set $\Gamma$, the disagreement update of $\model$ is defined as in Definition~\ref{def:prioritisations}. 
Then, satisfiability of a formula $A \disless{\Gamma} B$ at a world of a (weakly) centered model is defined as follows:
\begin{itemize}
	\item 	$\model, x \modelsm{\mathsf d} A \disless{} B $ \emph{iff} for all $ \alpha \in \sph{x}$, if $ \alpha \fem{\mathsf d} B$, then $ \alpha \fem{\mathsf d} A$;
	\item $\model, x \modelsm{\mathsf u} A \disless{\Gamma} B $ \emph{iff} $ \upmod{\Gamma}{x}{\mathsf{d}}, x \modelsm{\mathsf d} A \disless{} B$. 
%	\label{ap sat:cp:less}
\end{itemize}
%\begin{eqnarray}
%	\model, x \modelsm{\mathsf u} A \disless{} B & \text{ iff } & \text{for all } \alpha \in \sph{x}, \text{ if } \alpha \fem{\mathsf u} B, \text{ then } \alpha \fem{\mathsf u} A
%	\label{ap sat:lew:less}
%	\\[0.1cm]
%	\model, x \modelsm{\mathsf u} A \disless{\Gamma} B &\text{ iff } & \upmod{\Gamma}{x}{\mathsf{u}}, x \modelsm{\mathsf u} A \disless{} B
%	\label{ap sat:cp:less}
%\end{eqnarray} 
Again, the  clause for $\less$ corresponds to Lewis' satisfiability condition for the comparative plausibility at nested sphere models \cite{lewis1973}. Lewis showed that the comparative plausibility and the counterfactual are interdefinable, and this result easily extends to our framework. We omit the proof, which is routine from~\cite{lewis1973}. 

\begin{lemma} 
	\label{lemma:eq}
For $\model$ (weakly) centered sphere model  and $x$ world, it holds that: 
	\begin{enumerate}[noitemsep]
		\item \label{ap it:lemma:eq:cond}For \!\! $A, B \in  \Lcp$, 
		$\model, x \modelsm{\mathsf{d}} A \discond{\Gamma} B $ iff $ \model, x \modelsm{\mathsf{d}} (\bot \disless{\Gamma} A) \lor \lnot \big((A \land \lnot B) \disless{\Gamma} (A \land B) \big)$\\[-0.2cm]
		\item For $A, B \in \Lcpp$, 
		$\model, x \modelsm{\mathsf{d}} A \disless{\Gamma} B $ iff $ \model, x \modelsm{\mathsf{d}}  \big((A \lor B) \discond{\Gamma} \bot \big) \lor \big( (A \lor B) \discond{\Gamma} \lnot A \big)$  
	\end{enumerate}
\end{lemma}

% \begin{proof}[sketch]
% 	The proofs are routine, hence we only discuss one direction of $1)$. Suppose $\model, x \modelsm{\mathsf u} A\discond{\Gamma} B$. If $\Gamma = \varnothing$, the proof proceed as in \cite{lewis1973}. 
% 	Otherwise, by definition, $\upmod{\Gamma}{x}{\mathsf{u}}, x \modelsm{\mathsf{u}} A \discond{} B$. We need to show that either $\upmod{\Gamma}{x}{\mathsf{u}}, x \modelsm{\mathsf{u}} \bot \disless{} A$, or $\upmod{\Gamma}{x}{\mathsf{u}}, x \not \modelsm{\mathsf{u}} (A \land \lnot B) \disless{} (A \land B)$. Again, the proof proceeds as in \cite{lewis1973}. 
% \end{proof}
%\marianna{introduce notation for $\validm{}^{\VL}A $ and $\validm{}^{\VL^\mathsf{d}} A$}
%\marianna{added this}
In light of \cref{lemma:eq}, with an abuse of notation we say that, for $A \in \Lcp$,  $A$ is \emph{valid in weakly centered} (resp. \emph{centered}) \emph{models models under the disagreement update}, in symbols 
$\validm{}^{\mathsf{VW}^\mathsf{d}} A$ (resp. $\validm{}^{\mathsf{VC}^\mathsf{d}} A$), iff  $\model, x \modelsm{\mathsf{d}} A$ holds  for all worlds $x$ and all weakly centered (resp. centered) models $\model$. 
%Similarly, for $A \in \Llp$, we write $\validm{}^{\VL} A$  to denote validity in Lewis' logics $\VW $ and $\VC$ respectively. 
Similarly, for $A \in \Llp$, we write $\validm{}^{\VW} A$ and $\validm{}^{\VC} A$ to denote validity in Lewis' logics $\VW $ and $\VC$ respectively. 
The following immediately holds, since satisfaction of $A \in \Llp$ within $\VL$ and $\VLu{d}$ remains unchanged, as no updates are performed in $\VLu{d}$. 

%\marianna{yes, but I wanted to be explicit there. I put it back}
\begin{fact}
    \label{lemma:from_less_to_lesscp}
    For $A \in \Llp$, it holds that $\validm{}^{\VL}A $ iff $\validm{}^{\VL^\mathsf{d}} A$.
\end{fact}

% Soundness can be immediately  extended to $\VWu{d}$ and $\VCu{d}$, 
%  as the truth condition to evaluate $F$ within $\VWu{d}$ and $\VCu{d}$ is the same as in Lewis' semantics, as no updates are performed. 

% \begin{theorem}[Soundness] 
% 	%For any choice of  $\mathsf{u} \in \{\mathsf{i}, \mathsf{a}, \mathsf{d} \}$, and f
% 	For any formula $F \in \Lcpp$, if $\vdash_{\VW} F$ (resp. $\vdash_{\VC} F$ ) then $\modelsm{\mathsf{d}}^{\mathsf{VW}} F$ (resp. $\modelsm{\mathsf{d}}^{\mathsf{VC}} F$). 
% \end{theorem}
%\begin{proof}[sketch]
%	Let $\vdash_{\VW} F$ (resp. $\vdash_{\VC} F$). Then, $F \in \Llp$, and by \Cref{ap thm:lewis}, $F$ is valid in weakly centered (resp. centered) sphere models. The 
%	truth condition to evaluate $F$ within $\VWu{u}$ and $\VCu{u}$ is the same as in Lewis' semantics, as no updates are performed. Thus, $F$ remains valid  in $\VWu{u}$ (resp. $\VCu{u}$). 
%\end{proof}

Next, we investigate the relation of \emph{cp}-comparative plausibility formulas with their correspondents in Lewis' language $\Llp$, 
%To prove a version  of Fact~\ref{lemma:from_less_to_lesscp} for $\Lcpp$-formulas, 
aiming at translating $\Lcpp$ formulas into equisatisfiable $\Llp$ formulas. 
%The corresponding $\Llp$ formula is constructed by looking at a model $\model$ and a world $x$. Thus, 
The construction is quite complex, and it depends on the model chosen to evaluate $\Lcpp$ formulas. 
% and update.
 %We 
%shall explicitly consider the case of the disagreement update, and then discuss completeness of the logics which employ the implausibility and agreement updates. 
%We  introduce some additional definitions.

%\todoavg{23/02 added def of forcing set, since it does not appear before now}
\begin{definition}
	For any model $\model$, world $x$ and $\Gamma$ \emph{cp}-set of formulas, the \emph{forcing set of $\Gamma, x$} is defined as 
	$\gtruthset{x}{\Gamma} = \{G \in \Gamma \mid x \modelsm{\mathsf{d}} G \}$ and its complement as $\gtruthset{x}{\Gamma}^c=  \{G \in \Gamma \mid x \not\modelsm{\mathsf d} G \}$. Moreover,  let the \emph{set of paired subsets of $\Gamma$}  be $\pr\Gamma = \{\lambda \subseteq\Gamma \mid \lambda \text{ is paired}\}$. 
\end{definition}

For $A \in \Lcpp$, let $\cplless{A}$ denote the number of \emph{cp}-connectives with non-empty \emph{cp}-sets occurring in $F$.  
If $A \in \Lcpp$ and  $\cplless{A} = 0$, then   $A \in \Llp$. 
%With a slight abuse of notation, %in what follows 
%we shall write $\model, x \models A$ instead of $\model, x \modelsm{\mathsf d} A$ in case $\cplless{A}= 0$. 
The proof of the following Lemma can be found in the Appendix.

\begin{lemma}
	\label{lemma:paired:dis}
Take $F = A \disless{\Gamma} B \in \Lcpp$,  with 
%$\cplless{F} = 1$ and $\cplless{A} = \cplless{B} =0$ (thus, $\Gamma \neq \varnothing$). 
$\Gamma \neq \varnothing$. 
For $\model$  (weakly) centered model and world $x$,  $\model, x \modelsm{\mathsf{d}} F$ iff  $\model, x \modelsm{\mathsf{d}} \hat F$, where $\hat F$ is the formula:     
	\begin{eqnarray} 
		\hat F &  =&    \bigwedge_{\lambda \in \pr\Gamma } 
		\bigg[^1    
		\bigg(^2 
		\bigwedge_{\pr \Gamma \ni \lambda' \lexlefts x \lambda} 
		\neg \Diamond 
		\big( 
		A \land  \bigwedge\gtruthset{x}{\Gamma \setminus \lambda'} 
		\big)
		\bigg)^2
		\to     
        \nonumber\\
		& &
		\to \bigg (^3 
		\big (
		\bigvee_{
			\pr\Gamma \ni \lambda^{''}\lexleft x \lambda
		}   (A \land \bigwedge\gtruthset{x}{\Gamma \setminus \lambda''} )
		\big )
		\less
		\big (
		B \land \bigwedge\gtruthset{x}{\Gamma \setminus \lambda} 
		\big )
		\bigg)^3 \,
		\bigg]^1.  
        \label{eq:big_formula}
	\end{eqnarray}
\end{lemma}

If $\cplless{A} = \cplless{B} = 0$, then $\hat F \in \Llp$; else, $\hat F \in \Lcpp$. 
Due to Definition \ref{def:satisfaction}, 
the statement is well defined.
%\marianna{24 modified this as discussed}
%\todoavg{24 added ref of def}
Intuitively, formula $\hat F$ `describes' the updated sphere model using formulas in $\Llp$. 
The paired sets  $\lambda \in \pr{\Gamma}$ represent the disagreement sets between worlds in $\bigcup \sph{x}$ and the actual world $x$. Thus, each $\lambda$ can be thought of as representing a degree of disagreement w.r.t. the actual world. We can thus say that, e.g., $\disset{x}{u}{\Gamma} \lexleft x \lambda $, for some $u \in \bigcup \sph x$. 
Informally, whenever this happens we will say that ``$u$ disagrees from $x$ less or equal to $\lambda$''. 
For some $ z \in \bigcup S$ and $\lambda' \in \pr{\Gamma}$, $z \models \bigwedge\val{x}_{\Gamma \setminus \lambda'} $ means
%, by \cref{prop:compl} that $\Gamma\setminus \lambda'\subseteq \agset{x}{v}{\Gamma}$, i.e., 
that  $z$ satisfies all the formulas in  $\Gamma \setminus \lambda'$ which are also satisfied in $x$. % whence it satisfies their conjunction. 
%\footnote{Also refer to \cref{prop:compl} in the Appendix.}. 
Consequently, the set of formulas $z$ and $x$ disagree upon is a subset of $\lambda'$, from which we obtain   $\disset{x}{z}{\Gamma} \lexleft x \lambda' $. Considering worlds that satisfy  (conjunction of) formulas within $\Gamma \setminus \lambda'$ which $x$ satisfies then corresponds to selecting a world which
disagrees from $x$ less or equal to $\lambda'$. 
%whose disagreement set is lesser or equal to $\lambda'$. This allows to establish a relation between satisfiable formulas in $ \Lcpp$ and $ \Llp$.
%
Then, $\hat F$ is a conjunction of formulas, one for each $\lambda \in \pr\Gamma$. Let us fix an arbitrary such $\lambda$, corresponding to the disagreement set $\disset{x}{u}{\Gamma}$ of some world $u \in \bigcup \sph x$.  
The formula within parentheses 2 in \eqref{eq:big_formula} above states that  for any choice of $\lambda' \in \pr \Gamma$ such that $\lambda '\lexlefts x \lambda$,
	there is no world 
 that disagrees less or equal to $\lambda'$
	with $x$ and forces $A$. 
%	Thus, there is no world 
   % forcing $A$ and whose disagreement set w.r.t. $x$ and $\Gamma$ is $\lexlefts x$-smaller than $\lambda$.  
	The formula in parenthesis 3 says that if there is a world $v\in \alpha \in \sph{x}$ 
    %whose disagreement set w.r.t. $x$ is $\lexleft x$-smaller than $\lambda$
    that disagrees less or equal  to $\lambda$ with $x$ 
    and that forces $B$, then there exists a world $y \in \alpha$ 
    which disagrees less or equal to 
    $\lambda$ with $x$ 
    %whose disagreement set w.r.t. $x$ is $\lexleft x$-smaller than $\lambda$ 
    and that forces $A$. This corresponds to the truth condition of the $\less $ operator. 
    %where the specific sets of formulas which the worlds force ensure that we are considering the correct worlds w.r.t. the disagreement sets. 
    Then, 
    for each choice of $\lambda$, either 
	there is a world $u'$  which satisfies ${A}$ and such that $\disset{x}{u'}{\Gamma}\lexleft x \disset{x}{u}{\Gamma}$, thus making 
	the formula in parenthesis 2 false or, if no such $u'$ exists, 
	then the formula in 3 needs to be satisfied. 

Next, we generalise  \Cref{lemma:paired:dis} to arbitrary formulas $F$ of $ \Lcpp$ and  conclude with soundness and completeness results. The proofs of both statements are in the Appendix.

\begin{proposition}
\label{prop:rec:d}
For  $A \in \Lcpp$, $\model$ (weakly) centered sphere model and $x$ world, it holds that $\model, x \modelsm{\mathsf{d}} A$ iff  $\model, x \models \hat A$. 
\end{proposition}

\begin{theorem}%[Soundness and Completeness]
    \label{thm:d_iff_less}
    For $A \in \Lcp$, $\validm{}^{\VL^\mathsf d} A$ iff $\validm{}^{\VL} \hat A$. 
\end{theorem}

\begin{figure}[t]
    \begin{center}
    \begin{tikzpicture}
        \draw[fill= gray!20] (-2, 0.4) ellipse (0.9cm and 0.7cm);
        \draw (2, 0)[]  ellipse (1.1cm and 1.1cm);
        \draw[fill= gray!20] (2, 0.4) ellipse (0.9cm and 0.7cm);
        \draw[fill= gray!10] (2.5, 0.4) ellipse (0.4cm and 0.5cm);
        \node[] at (-2, 0.4) (Llp) {\small$\Llp$};
        \node[] at (3.3, -0.6) (Lcpp) {\small$\Lcpp$};
        \node[] at (2.5, 0.4) (hat) {$ \hat A$};
        \node[] at (1.7, 0.4) (Lzero) {$ \Llp_{\varnothing}$};
        \node[] at (2,-0.7) (A) {$ A$};
        \draw[<->,densely dotted] (A) to [out=0,in=-90] (hat);
        \draw[<->,dashed] (Llp) to [out=20,in=160] (Lzero);
    \end{tikzpicture}
    \vspace{-0.4cm}
    \end{center}
    \caption{
    Relationship between sets of valid formulas. The circle labelled with $\Llp$ (resp. $\Lcpp$) represents the set of formulas of $\Llp$ (resp. $\Lcpp$) valid at (weakly) centred sphere models. $\Llp_\varnothing$ represents the subset of  valid formulas of $\Lcpp$ having empty \emph{cp}-sets. This set coincides with the set of valid formulas in $\Llp$ (Fact~\ref{lemma:from_less_to_lesscp}, dashed arrow). 
    The formulas $A \in \Lcpp \setminus \Llp_\varnothing$ are s.t.  $\cplless{A} >0$. 
    By \cref{thm:d_iff_less} (dotted arrow) every valid formula in $\Lcpp$ with non-empty \emph{cp}-sets is mapped into a fragment of the valid formulas of $\Llp_\varnothing$. 
    }
    \label{fig:relationship}
\end{figure}

%Observe that $\Llp$ identifies a subset of the formulas in $\Lcpp$. 
%\marianna{please check if next paragraph makes sense. I first thought $\hat A$ identifies a fragment of formulas in $\Llp$, but this is not true!}

%\todoavg{23/02 i am just confused}

%\marianna{this probably deserves a drawing or more explanation, I'll try to add it tomorrow}
By definition, the formula $\hat A$ is either a formula of $\Llp$ (in case all its \emph{cp}-sets are empty) or a formula expressible in the language of $\Llp$. Thus, by 
Fact~\ref{lemma:from_less_to_lesscp} and 
\cref{thm:d_iff_less}, the set of valid formulas of $\VLu d$ coincides with the set of valid formulas of $\VL$ (refer to \cref{fig:relationship}). Since Lewis defined axiom systems for logics $\VL$ in~\cite{lewis1973}, \cref{thm:d_iff_less} also provides us with an axiomatization for logics $\VLu d$. 

\section{Conclusions and future work}
\label{sec:conclusions}
In this paper we introduced an innovative way of evaluating \emph{cp}-counterfactuals, by updating the worlds of a sphere model according to the significance of formulas, which we measured by calculating the weights of specific sets of formulas  associated to each world in a system of spheres. The disagreement update, performed on (weakly) centered sphere models, gives rise to  logics $\VWu{d}$, $\VCu{d}$, which we showed to be complete w.r.t. Lewis' logics $\VW$ and $\VC$.

We plan to extend our analysis of \emph{cp}-counterfactuals 
% as dynamic operators, possibly defining new updates and testing our updates on  wider families of conditional logics than those considered in this paper. 
%We plan to investigate 
%by studying \emph{cp}-modalities 
in various directions. 
%models without nesting. 
%Such a framework would allow to define \emph{cp}-preferential modalities. 
We wish to deepen the study of the comparative plausibility operator $\disless{\Gamma}$ (following~\cite{dalmonte:comparative22}, in non-nested models), and compare it to the preferential \emph{cp}-operator introduced in \cite{van2009everything}.
Furthermore, we wish to explore additional  constraints on \emph{cp}-sets, such as, e.g.,  allowing \cpcounts to occur within \emph{cp}-sets, or requiring that \emph{cp}-sets are consistent or closed under subformulas.   
Finally, we plan to include impossible worlds in our account,  following ideas from~\cite{weiss2017semantics}. Our approach would allow to distinguish between different kinds of impossible worlds, allowing for a nuanced  evaluation of counterfactuals within impossible states.

%Specifically, our approach allows to distinguish between different kinds of impossible worlds (e.g., empirically, physically or logically impossible worlds). Then, much like we have done in this work, we could use the notions of significance and weight to evaluate counterfactual statements also within impossible states.

%There are several ways in which impossible worlds could be incorporated into our account.
%For example, we could distinguish between different kinds of impossible worlds (e.g., empirically, physically or logically impossible worlds). Then, much like we have done in this work, we could use the notions of significance and weight to evaluate counterfactual statements also within impossible states.

\clearpage

% ---- Bibliography ----
%
% BibTeX users should specify bibliography style 'splncs04'.
% References will then be sorted and formatted in the correct style.

 \bibliographystyle{splncs04}
 \bibliography{cp-bib.bib}

\begin{thebibliography}{10}
\providecommand{\url}[1]{\texttt{#1}}
\providecommand{\urlprefix}{URL }
\providecommand{\doi}[1]{https://doi.org/#1}

\bibitem{burgess1981quick}
Burgess, J.: Quick completeness proofs for some logics of conditionals. Notre
  Dame Journal of Formal Logic  \textbf{22} (01 1981),
  \url{https://doi.org/10.1305/ndjfl/1093883341}

\bibitem{dalmonte:comparative22}
Dalmonte, T., Girlando, M.: Comparative plausibility in neighbourhood models:
  axiom systems and sequent calculi. In: Fern{\'{a}}ndez{-}Duque, D.,
  Palmigiano, A., Pinchinat, S. (eds.) Advances in Modal Logic, AiML 2022. pp.
  305--327. College Publications (2022),
  \url{http://www.aiml.net/volumes/volume14/20-Dalmonte-Girlando.pdf}

\bibitem{this-is-fine}
Fine, K.: Critical notice of \uppercase{L}ewis, \uppercase{C}ounterfactuals.
  Mind  \textbf{84}(335),  451--458 (1975),
  \url{https://www.jstor.org/stable/2253565}

\bibitem{DBLP:journals/corr/GirardT16}
Girard, P., Triplett, M.A.: Ceteris paribus logic in counterfactual reasoning.
  In: Ramanujam, R. (ed.) Proceedings Fifteenth Conference on Theoretical
  Aspects of Rationality and Knowledge, {TARK} 2015, Carnegie Mellon
  University, Pittsburgh, USA, June 4-6, 2015. {EPTCS}, vol.~215, pp. 176--193
  (2015), \url{http://www.tark.org/proceedings/tark_jun4_15/TARK2015.13.pdf}

\bibitem{girard2018prioritised}
Girard, P., Triplett, M.A.: Prioritised ceteris paribus logic for
  counterfactual reasoning. Synthese  \textbf{195}(4),  1681--1703 (2018),
  \url{https://doi.org/10.1007/s11229-016-1296-5}

\bibitem{goodman1947problem}
Goodman, N.: The problem of counterfactual conditionals. The Journal of
  Philosophy  \textbf{44}(5),  113--128 (1947),
  \url{https://www.jstor.org/stable/2019988}

\bibitem{kaufmann2013causal}
Kaufmann, S.: Causal premise semantics. Cognitive science  \textbf{37}(6),
  1136--1170 (2013), \url{https://doi.org/10.1111/cogs.12063}

\bibitem{kratzer1981partition}
Kratzer, A.: Partition and revision: The semantics of counterfactuals. Journal
  of Philosophical Logic  \textbf{10},  201--216 (1981),
  \url{https://doi.org/10.1007/BF00248849}

\bibitem{lewis1973}
Lewis, D.: Counterfactuals. Blackwell, Oxford (1973)

\bibitem{lewis1979counterfactual}
Lewis, D.: Counterfactual dependence and time's arrow. No{\^u}s
  \textbf{13}(4),  455--476 (1979), \url{https://www.jstor.org/stable/2215339}

\bibitem{priest2008introduction}
Priest, G.: An Introduction to Non-Classical Logic: From If to Is. Cambridge
  Introductions to Philosophy, Cambridge University Press, Cambridge, 2 edn.
  (2008). \doi{10.1017/CBO9780511801174}

\bibitem{van2009everything}
Van~Benthem, J., Girard, P., Roy, O.: Everything else being equal: A modal
  logic for ceteris paribus preferences. Journal of philosophical logic
  \textbf{38}(1),  83--125 (2009),
  \url{https://doi.org/10.1007/s10992-008-9085-3}

\bibitem{van2007dynamic}
Van~Ditmarsch, H., van Der~Hoek, W., Kooi, B.: Dynamic epistemic logic,
  vol.~337. Springer Science, Dordrecht (2007),
  \url{https://doi.org/10.1007/978-1-4020-5839-4}

\bibitem{von1963logic}
Von~Wright, G.H.: The logic of preference, vol.~40. Edinburgh University Press
  (1963)

\bibitem{von1972logic}
Von~Wright, G.H.: The logic of preference reconsidered. Theory and Decision
  \textbf{3}(2) (1972), \url{https://doi.org/10.1007/BF00141053}

\bibitem{weiss2017semantics}
Weiss, Y.: Semantics for counterpossibles. The Australasian Journal of Logic
  \textbf{14}(4),  383--407 (2017),
  \url{https://doi.org/10.26686/ajl.v14i4.4050}

\end{thebibliography}

\clearpage
% ---- Appendix ----
\appendix
\section{Appendix}

% \subsection*{Definition of the \emph{cp}-languages}

% Following~\cite{seligman2011flexibility,girard2018prioritised} we here present the full definition of the language $\Lcp$. 
% \begin{definition}
% \label{def:full:language}
%     Fix a countable set of propositional atoms $\atm = \{ p_0, p_1, \dots\}$. 
%     For each natural number $n$, let $\Lcp({n})$ be defined by the following grammar:
%     $$
%     A:: = p \mid \bot \mid A \rightarrow A \mid A \discond{\Gamma} A.
%     $$
%     In the above, $p\in \atm$ and $\Gamma   $ is a finite set of $\Lcp(m)$ formulas, for $m < n$. Moreover, we take $\Gamma$ to be \emph{paired}, that is: $F \in \Gamma $ iff $F \rightarrow \bot \in \Gamma$ (we shall often abbreviate $F \rightarrow \bot$ with $\bar F$ or $\lnot F$). 
%     The language $\Lcp$ is taken to be $\bigcup_n \Lcp(n)$. 
% \end{definition}

% The language $\Lcpp$ from \cref{sec:completeness} is defined similarly. 

\subsection*{Proofs from \cref{sec:weight}}

\begin{proof}[of \cref{prop:Symmetry}]
Let $\weight{x}{A}=(w^{ A}_0, \dots, w^{ A}_n)$,   $\weight{x}{ B}=(w^{ B}_0, .., w^{ B}_n)$,  $\weight{x}{\bar A}=(w^{\bar A}_0, .., w^{\bar A}_n)$ and  $\weight{x}{\bar B}=(w^{\bar B}_0, .., w^{\bar B}_n)$. 
From Definition~\ref{weight_of_formulas} it follows that, since worlds are classical, for any $l \leq n$ it holds that $w^{\bar A}_l < w^{\bar B}_l$ if and only if  $w^{ B}_l < w^{ A}_l$. 
Now suppose $B \xellefts x A$. Then, for the first $k \leq n$ such that $w^A_k \neq w^B_k $, it holds that  $w^{ A}_k < w^{ B}_k$. Thus, $w^{ \bar B}_k < w^{ \bar A}_k$, and  $\bar A \xellefts x \bar B$. If $B \eq x A$, then for all $k \leq n$, $w^A_k = w^B_k $, and thus $w^{\bar A}_k = w^{\bar B}_k $, hence $\bar A \eq x \bar B$. 
\qed
\end{proof}

\begin{proof}[of \ref{monotonicity}]
    We only cover the case when $\Gamma \subset \Delta$, whence  $\Delta \setdiff \Gamma \neq \varnothing$. Then, 
    $\weight{x}{\Gamma} = (\weight{x}{c_1},.., \weight{x}{c_n} )$ and  
    $\weight{x}{\Delta} = (\weight{x}{d_1},.., \weight{x}{d_{n+l}} )$, for some $0 < l$. 
    If there is some $A \in \Delta \setdiff \Gamma$ such that for some  $B \in \Gamma$ it holds that 
    $ B \xellefts x A$, 
    then by definition $\Gamma \lexlefts x \Delta $. 
    If instead for all formulas $A \in \Delta \setdiff \Gamma$ and all formulas $B \in \Gamma$ it holds $ A\xelleft x B$, then we conclude that $\Gamma \lexlefts x \Delta $ because $n < n+l$. 
    \qed 
\end{proof}

\subsection*{Proofs from \cref{sec:update}}

\begin{proof}[of \cref{update properties}]
Non-emptiness and nesting of $\upsph{d}{x}$ follow from Definition~\ref{def:update}.  
Moreover, $\weight{x}{\disset{x}{x}{\Gamma}}=\min\{\weight{x}{\disset{x}{y}{\Gamma}}\mid y \in \bigcup \sph{x}\} $. Thus, if $\model$ is weakly centered, for all $\alpha \in\ups{\Gamma}{x}{\mathsf{d}}(x)$ we have $x \in \alpha$, hence $\upmod{\Gamma}{x}{\mathsf{ d}}$ is weakly centered. If $\model$ is centered, then by definition
$\{x\} \in \ups{\Gamma}{x}{\mathsf{d}}(x)$. So $\upmod{\Gamma}{x}{\mathsf{d}}$ is centered. 
\end{proof}

\subsection*{Proofs from \cref{sec:completeness}}

The following notion of agreement set is the dual of the disagreement set, which we took as primitive to define our updates (Definition~\ref{def:prioritisations}). The following Lemma explicit their duality. Both the definition and the lemma will be used in the proof of \cref{lemma:paired:dis}. 

\begin{definition}
For $\model = \langle W, S, V\rangle$, worlds $x,y \in W$ and a \emph{cp}-set of formulas $\Gamma$, the \emph{agreement set of $y$ w.r.t. $\Gamma$, $x$} is the set $\agset{x}{y}{\Gamma} = \{G \in \Gamma \mid x \modelsm{} G \textit{ iff } y \modelsm{} G \}$. 
\end{definition}

\begin{lemma}
	\label{agree to disagree}
	For a (weakly) centered 
	model $\model$, worlds $x, y,z \in \bigcup \sph x$ and \emph{cp}-set $\Gamma$, it holds that:
	\begin{enumerate}
		\item \label{1st clause}$\Gamma = \agset{x}{y}{\Gamma} \cup \disset{x}{y}{\Gamma}$;
		\item \label{2nd clause} $\disset{x}{y}{\Gamma}\lexleft x \disset{x}{z}{\Gamma} \,  \text{iff} \, \agset{x}{z}{\Gamma} \lexleft x \agset{x}{y}{\Gamma}$.
	\end{enumerate}
\end{lemma}

\begin{proof}
	\Cref{1st clause} is an immediate consequence of the definition of agreement and disagreement sets, considering that $\Gamma$ is paired. We also note that by definition, $\agset{x}{y}{\Gamma}\cap \disset{x}{y}{\Gamma}=\varnothing$.
	To prove one direction of \Cref{2nd clause},  assume that $ \weight{x}{\agset{x}{z}{\Gamma}} \prel \weight{x}{\agset{x}{y}{\Gamma}}$. 
	We shall first prove that:
	\begin{eqnarray}
		\label{a} 
		\weight{x}{\agset{x}{z}{\Gamma}\setminus \agset{x}{y}{\Gamma}} \prel \weight{x}{\agset{x}{y}{\Gamma}\setminus \agset{x}{z}{\Gamma}}
	\end{eqnarray}
	If $\agset{x}{y}{\Gamma}\cap \agset{x}{z}{\Gamma} = \varnothing$ we immediately obtain that \eqref{a} holds. 
	Otherwise, there are formulas $B_1, \dots, B_n \in \agset{x}{z}{\Gamma} \cap  \agset{x}{y}{\Gamma} $. Removing  $\weight{x}{B_1}, \dots, \weight{x}{B_n}$ from both lists $\weight{x}{\agset{x}{z}{\Gamma}} $ and $\weight{x}{ \agset{x}{y}{\Gamma}}$ does not alter the order w.r.t. $\prel$, as we remove the same elements
	from both lists. 
	Thus, $ \weight{x}{\agset{x}{z}{\Gamma}} \setminus (\weight x B_1,\dots, \weight x B_n)  \prel \weight{x}{\agset{x}{y}{\Gamma} }\setminus (\weight x B_1,\dots, \weight x B_n) $, whence we conclude that   \eqref{a} holds. 
	
	From \Cref{1st clause} it easily follows that $(\agset{x}{y}{\Gamma}\setminus \agset{x}{z}{\Gamma}) = (\disset{x}{z}{\Gamma}\setminus \disset{x}{y}{\Gamma}) $ and   $(\agset{x}{z}{\Gamma}\setminus \agset{x}{y}{\Gamma})= (\disset{x}{y}{\Gamma} \setminus \disset{x}{z}{\Gamma}) $. From these and \eqref{a} we obtain: 
	\begin{eqnarray}
		\label{b}   \weight{x}{\disset{x}{y}{\Gamma} \setminus \disset{x}{z}{\Gamma} } \prel \weight{x}{ \disset{x}{z}{\Gamma}\setminus \disset{x}{y}{\Gamma}}
	\end{eqnarray}
	To conclude the proof, we need to show that the following holds:
	\begin{eqnarray}
		\label{c}     \weight{x}{\disset{x}{y}{\Gamma}  } \prel \weight{x}{ \disset{x}{z}{\Gamma}}
	\end{eqnarray}
	If $\disset{x}{y}{\Gamma} \cap \disset{x}{z}{\Gamma}=\varnothing$, then
	\eqref{c} immediately follows from \eqref{b}. Otherwise, there are  formulas $C_1, \dots, C_k \in \disset{x}{z}{\Gamma}\, \cap \, \disset{x}{y}{\Gamma} $. Adding  $\weight{x}{C_1}, \dots, \weight{x}{C_k}$ in both lists $\weight{x}{\disset{x}{z}{\Gamma}} $ and $\weight{x}{ \disset{x}{y}{\Gamma}}$ does not alter the order w.r.t. $\prel$, as we add the same number of elements to each list, with the same value. Thus, from \eqref{c} we obtain that $(\weight{x}{\disset{x}{y}{\Gamma} \setminus \disset{x}{z}{\Gamma} })\cup (\weight x C_1, \dots, \weight x C_k) \prel (\weight{x}{ \disset{x}{z}{\Gamma}\setminus \disset{x}{y}{\Gamma}}) \cup (\weight x C_1, \dots, \weight x C_k) $, from which \eqref{c} immediately follows. 
	
	The converse direction of Clause \ref{2nd clause} is proved similarly. 
    \qed
\end{proof}

\begin{proposition}
	\label{prop:compl}
	For any (weakly) centered sphere $\model$, and $x, v \in \bigcup \sph{x}$, 
	$\Gamma$ \emph{cp}-set and $\lambda \in \pr \Gamma$, it holds that:
	\begin{enumerate}
		
		\item \label{ap 1st}$ 
		\agset{x}{v}{\Gamma}=(\gtruthset{x}{\Gamma} \cap \val{v}_\Gamma)\cup (\gtruthset{x}{\Gamma}^c\cap \val{v}_\Gamma^c)$; 
		\item \label{ap 2nd}
		$\disset{x}{v}{\Gamma}=(\gtruthset{x}{\Gamma} \cap \val{v}_\Gamma^c)\cup (\gtruthset{x}{\Gamma}^c\cap \val{v}_\Gamma)$;
		\item \label{ap 3rd} 
		$\model , v\modelsm{\mathsf{d}}   \, \bigwedge \gtruthset{x}{\Gamma \setminus \lambda} $ \, if and only if \, $\disset{x}{v}{\Gamma}\subseteq \lambda$;
		\item \label{ap 4th} 
		$\model, v\modelsm{\mathsf{d}}  \bigwedge \gtruthset{x}{\lambda} $ \ if and only if $\lambda\subseteq \agset{x}{v}{\Gamma}$;
		\item \label{ap 5th} 
		if \, $\disset{x}{v}{\Gamma}\subseteq \lambda$ 
		\, then \,  $\weight{x}{\disset{x}{v}{\Gamma}} \prel \weight{x}{\lambda} $.
		
	\end{enumerate}
\end{proposition}

\begin{proof}[sketch]
	\Cref{ap 1st} states that the agreement set of $x$ and $v$ in $\Gamma$ corresponds to the set of formulas in $\Gamma$ which $x$ and $v$ both satisfy, plus the set of formulas in $\Gamma$ that both $x$ and $v$ do not satisfy. 
	Similarly, \Cref{ap 2nd} defines the disagreement set of $x$ and $v$ in $\Gamma$ as the formulas in $\Gamma$ that $x$ satisfies but $v$ does not, plus the set of formulas in $\Gamma$ that $v$ satisfies, but $x$ does not. Both statements  follow from  Definition~\ref{weight of sets} and the definition of $\gtruthset{x}{\Gamma}^c$. 
	To prove \Cref{ap 3rd}, observe that $v\modelsm{\mathsf{d}}  \bigwedge \gtruthset{x}{\Gamma \setminus \lambda} 
	\ $ if and only if for any formula $G \in \Gamma \setminus \lambda$, if $x\modelsm{\mathsf{d}}  G$ then $v \modelsm{\mathsf{d}}  G$ (and also if $x \not \modelsm{\mathsf{d}}  G$ then $v \not \modelsm{\mathsf{d}}  G$). 
	Since $x$ and $v$ assign the same value to formulas in $\Gamma \setminus \lambda$ then,  by Definition~\ref{def:prioritisations}, $\Gamma \setminus \lambda =\agset{x}{v}{\Gamma\setminus \lambda} $. Since  agreement and disagreement sets are complementary (\cref{agree to disagree}), we obtain $ \disset{x}{v}{\Gamma\setminus \lambda}= \varnothing$, which is equivalent to  $ \disset{x}{v}{\Gamma} \subseteq \lambda$. 
	The proof of \Cref{ap 4th} is similar to that of \eqref{ap 3rd} and \Cref{ap 5th} is a special case of \cref{monotonicity}.   
    \qed
\end{proof}

\begin{proof}[of \Cref{lemma:paired:dis}]
	For one direction, assume $\model, x \modelsm{\mathsf{d}} A \disless{\Gamma} B $. 
	Then, by definition, $ \upmod{\Gamma}{x}{\mathsf{d}}, x \modelsm{\mathsf{d}} A \disless{} B$ which
	means that for all $\hat{\alpha} \in  \upds{\Gamma}{x}{d}$,
	if $\hat{\alpha} \fe B$, then $\hat{\alpha} \fe A$. 
	Let $\lambda \in \pr\Gamma$ such that: 
	\begin{equation}
		\label{ap  eqn:cpl:dir1}
		\model, x \models   
		\bigwedge_{\pr \Gamma \ni \lambda' \lexlefts x \lambda} 
		\neg \Diamond 
		\big( 
		A \land  \bigwedge\gtruthset{x}{\Gamma \setminus \lambda'}  
		\big)
	\end{equation}
	We need to show that $x$ satisfies the formula within parentheses 3.  
	Consider any world $v \in \bigcup \sph x$ such that $\model,v\models 
	B
	\land \bigwedge\gtruthset{x}{\Gamma \setminus \lambda} $ (observe that, if no such world exists, then the formula in parentheses 3 is vacuously satisfied, and we are done). 
	Then, $\model,v\models 
	B
	$ and $\model,v \models \bigwedge\gtruthset{x}{\Gamma \setminus \lambda} $. By \ref{ap 3rd}  of Proposition \ref{prop:compl} we conclude that $\disset{x}{v}{\Gamma} 
	\subseteq \lambda $. Therefore, by \cref{ap 5th} of Proposition \ref{prop:compl}, and by definition, we have that $\disset{x}{v}{\Gamma}
	\lexleft x \lambda $. 
	Since $\model,v \models 
	B
	$ and $v \in \bigcup \sph x$, then for some sphere $\hat \alpha \in \upsph{\mathsf{d}}{x}$ we have that $v \in  \hat \alpha \in \upds{\Gamma}{x}{d}$, and thus
	$\hat \alpha \fe 
	B
	$. Choose the smallest such sphere $\hat \alpha$. This sphere will contain only worlds $z$ such that $\disset{x}{z}{\Gamma} \lexleft x \disset{x}{v}{\Gamma}$. 
	By assumption, $\hat \alpha \fe A$, that is, there is a $y \in \hat \alpha$ such that $y \models A$, and   $\disset{x}{y}{\Gamma} \lexleft{x} \disset{x}{v}{\Gamma}$. 
	However, it cannot be that 
	$\disset{x}{y}{\Gamma} \lexlefts x \disset{x}{v}{\Gamma}$, because by taking $\lambda' = \weight{x}{\disset{x}{y}{\Gamma}}$ we would obtain that $x \models  \Diamond(
	A
	\land  (\bigwedge\gtruthset{x}{\Gamma \setminus \lambda'} )) $, contradicting (\ref{ap eqn:cpl:dir1}). Thus, we have that 
	$\disset{x}{y}{\Gamma}\eq{x}\disset{x}{v}{\Gamma}$.	Since $v \models B
	\land \bigwedge\gtruthset{x}{\Gamma \setminus \lambda}
	$ and $v \in \bigcup \sph x$, we also have that for some $\alpha \in \sph x$, $v \in \alpha $ and $ \alpha \fe 
	B
	\land \bigwedge\gtruthset{x}{\Gamma \setminus \lambda}$. 
	To conclude the proof, it remains to show that the following holds: 
	\begin{equation}
		\label{ap  eq:to_show}
		\alpha \fe \bigvee_{\pr\Gamma\ni \lambda^{''}\lexleft x\lambda}  
		(A \land \bigwedge\gtruthset{x}{\Gamma \setminus \lambda''}).
	\end{equation}
	Take $\lambda''=\disset{x}{y}{\Gamma}$. 
	Hence $y \models \bigwedge\gtruthset{x}{\Gamma \setminus \lambda''}$, by \ref{ap 3rd} of Proposition~\ref{prop:compl}. 
	Thus  $y \models 
	A
	\land \bigwedge \gtruthset{x}{\Gamma \setminus \lambda''}$, 
	from which we obtain that: 
	$$
	y \models \bigvee_{\pr\Gamma\ni \lambda^{''}\lexleft x\lambda}   
	(A \land \bigwedge\gtruthset{x}{\Gamma \setminus \lambda''}).
	$$
	Since $y \in \bigcup \upds{\Gamma}{x}{d} = \bigcup \sph x$, there is a $\beta \in \sph x$ such that $y \in \beta$. However, since $\disset{x}{y}{\Gamma}\eq x \disset{x}{v}{\Gamma}$
	and since $y,v \in \hat \alpha$, we have either $ \beta \subseteq \alpha$ or $\beta = \alpha$. In both cases we conclude that $y \in \alpha$, whence \eqref{ap eq:to_show} holds.

	 To prove the converse direction of the Lemma, assume that:
	\begin{align*}
		\model, x \models &
		\bigwedge_{\lambda \in \pr\Gamma } 
		\bigg[^1    
		\bigg(^2 
		\bigwedge_{\pr \Gamma \ni \lambda' \lexlefts x\lambda} 
		\neg \Diamond 
		\big( 
		A \land  \bigwedge\gtruthset{x}{\Gamma \setminus \lambda'}  
		\big)
		\bigg)^2
		\to     \\
		&\to \bigg (^3 
		\big (
		\bigvee_{
			\pr\Gamma \ni \lambda^{''}\lexleft x\lambda
		}  (A \land \bigwedge\gtruthset{x}{\Gamma \setminus \lambda''} )
		\big )
		\less
		\big (
		B \land \bigwedge\gtruthset{x}{\Gamma \setminus \lambda} 
		\big )
		\bigg)^3 \,
		\bigg]^1.  
	\end{align*}
	
	Let us take an arbitrary $\hat \alpha \in \upds{\Gamma}{x}{d}$ such that 
	$  \hat \alpha \fe B$, i.e., let us suppose that there is a $v \in \hat \alpha$ such that $v \models B$. We shall prove that $\hat \alpha \fe A$, i.e., we shall find an  $y \in \hat \alpha$ such that $y \models A$. We reason by case distinction. 
	First, suppose that for $\lambda = \disset{x}{v}{\Gamma}$. Then, it holds that: 
	$$
	\model, x \not \models \, 
	\bigwedge_{\pr \Gamma \ni \lambda' \lexlefts x\lambda} 
	\neg \Diamond 
	\big( 
	A \land  \bigwedge\gtruthset{x}{\Gamma \setminus \lambda'}  
	\big).
	$$
	This means that, for some $\lambda' \lexlefts \lambda$, there is a $y \in \bigcup \sph x$ such that $y \models A \land \bigwedge\gtruthset{x}{\Gamma \setminus \lambda'} $. Then, by \ref{ap 3rd} of \cref{prop:compl} we have $ \disset{x}{y}{\Gamma} \subseteq \lambda' $. We conclude that $\disset{x}{y}{\Gamma} \lexlefts x\disset{x}{v}{\Gamma}$. Thus, by Definition~\ref{def:update}, $y \in \hat \alpha$, from which we obtain  $\hat \alpha \fe A$, thus concluding the proof. Next, always for $\lambda = \disset{x}{v}{\Gamma}$ suppose: 
	$$
	\model, x \models \, 
	\bigwedge_{\pr \Gamma \ni \lambda' \lexlefts x\lambda} 
	\neg \Diamond 
	\big( 
	A \land  \bigwedge\gtruthset{x}{\Gamma \setminus \lambda'}  
	\big).
	$$
	Then, using our assumption, we obtain that:
	$$
	\model, x \models \, 
	\big (
	\bigvee_{
		\pr\Gamma \ni \lambda^{''}\lexleft x\lambda
	}  (A \land \bigwedge\gtruthset{x}{\Gamma \setminus \lambda''} )
	\big )
	\less
	\big (
	B \land \bigwedge\gtruthset{x}{\Gamma \setminus \lambda} 
	\big ).
	$$
	
	Since $\lambda = \disset{x}{v}{\Gamma}$, by \ref{ap 3rd} of \Cref{prop:compl} we have that $v\models \bigwedge\gtruthset{x}{\Gamma \setminus \lambda} $.
	Moreover, by assumption, $v \models B$, and therefore there is an $\alpha \in \bigcup \sph x$ such that $v \in \alpha $ and $\alpha \fe B \wedge (\bigwedge\gtruthset{x}{\Gamma \setminus \lambda} )$. Then, always by assumption, we have that:
	$$
	\alpha \fe 
	\big (
	\bigvee_{
		\pr\Gamma \ni \lambda^{''}\lexleft x\lambda
	}  A \land \bigwedge\gtruthset{x}{\Gamma \setminus \lambda''} 
	\big ). 
	$$
	Thus, for some $y \in \alpha$  we have that $y \models A \land \bigwedge\gtruthset{x}{\Gamma \setdiff \lambda''} $. 
	By Proposition~\ref{prop:compl}, we have that $ \disset{x}{y}{\Gamma} \subseteq \lambda''$ and $\disset{x}{y}{\Gamma} \lexleft x  \disset{x}{v}{\Gamma}$. 
	By construction (Definition~{def:update}) $y \in \hat \alpha$, and the proof is completed.
\end{proof}

% \begin{proof}[sketch, of \ref{prop:rec:d}]
%     If $A \in \Llp$, then the result holds by Fact~\ref{lemma:from_less_to_lesscp}. 
%     Otherwise, if $A \in \Llp$, we iterate the construction from \cref{lemma:paired:dis}. 
%     \qed 
% \end{proof}

\begin{proof}[sketch, of \ref{prop:rec:d}]
    If $\cplless{A} = 0$, set  $\hat A = A \in \Llp$, and the result immediately follows. 
    Otherwise, if $\cplless{A} = n > 0$, we `decompose' formula $A$, evaluating its satisfiability at $\model, x$ (and possibly other worlds in the model). Whenever we encounter a subformula $B \disless{\Gamma} C$ with $\Gamma \neq \varnothing$, we apply \cref{lemma:paired:dis}. After $n$ iterations of \cref{lemma:paired:dis}, we obtain a formula $\hat A \in \Llp$ equisatisfiable with $A$. \qed
    %Equisatisfiability immediately follows by \cref{lemma:paired:dis}.  
\end{proof}

\begin{proof}[of \cref{thm:d_iff_less}]
    To prove one direction, take an arbitrary $A\in \Lcpp$ such that  $\validm{}^{\mathsf{VL}^\mathsf d}A$. 
    Then, for arbitrary $\model$, $x$, it holds that $\model, x \modelsm{\mathsf d}  A$ and, by \cref{lemma:eq}, there is a formula $B \in \Lcpp$ such that $\model, x \modelsm{\mathsf d}  B$. 
    By \cref{prop:rec:d},  $\model, x \models \hat B$, for $\hat B \in \Llp$.  
    Then, since the above holds for arbitrary models,  conclude that $\validm{}^{\VL} \hat B$ and, again by \cref{lemma:eq}, we have that there is $\hat A \in \Lcp$ such that  $\validm{}^{\VL} \hat A$. The converse direction is proved similarly. 
    \qed
\end{proof}

\end{document}